\documentclass[a4paper,UKenglish,cleveref,autoref]{lipics-v2021}

\pdfoutput=1

\usepackage{xcolor}

\title{Contextual MetaML: Syntax and Full Abstraction}

\author{Haoxuan Yin}
{University of Oxford, United Kingdom \and \url{https://www.cs.ox.ac.uk/people/haoxuan.yin/}}
{haoxuan.yin@merton.ox.ac.uk}
{https://orcid.org/0009-0008-2817-0227}
{funded by Oxford-DeepMind Graduate Scholarship.}

\author{Andrzej S. Murawski}
{University of Oxford, United Kingdom \and \url{https://www.cs.ox.ac.uk/people/andrzej.murawski/}}
{andrzej.murawski@cs.ox.ac.uk}
{https://orcid.org/0000-0002-4725-410X}
{}

\author{C.-H. Luke Ong}
{Nanyang Technological University, Singapore \and \url{https://www3.ntu.edu.sg/home/luke.ong/}}
{luke.ong@ntu.edu.sg}
{https://orcid.org/0000-0001-7509-680X}
{}

\authorrunning{H. Yin, A.\,S. Murawski, and C.-H.\,L. Ong}

\Copyright{Haoxuan Yin, Andrzej S. Murawski, and C.-H. Luke Ong}

\ccsdesc[500]{Theory of computation~Program semantics}
\ccsdesc[300]{Software and its engineering~Imperative languages}
\ccsdesc[300]{Theory of computation~Type theory}

\keywords{Metaprogramming, operational game semantics, trace model, contextual modal type theory}

\relatedversiondetails[cite=yin2026layeredmodalmlsyntax]{Full version}{https://arxiv.org/abs/2602.03033}

\acknowledgements{We thank the reviewers for their helpful comments.
The work was partially supported by National Research Foundation, Singapore, under its RSS Scheme (NRF-RSS2022-009). 
}

\nolinenumbers

\EventEditors{Claudia Faggian and Joost-Pieter Katoen}
\EventNoEds{2}
\EventLongTitle{41st Annual Symposium on Logic in Computer Science (LICS 2026)}
\EventShortTitle{LICS 2026}
\EventAcronym{LICS}
\EventYear{2026}
\EventDate{July 20--23, 2026}
\EventLocation{Lisbon, Portugal}
\EventLogo{}
\SeriesVolume{380}
\ArticleNo{8}

\usepackage{mathtools}
\usepackage{multicol}
\usepackage{bussproofs}
\usepackage{listings/ocaml}
\usepackage{makecell}
\usepackage{booktabs}

\newenvironment{proof_sketch}[1][Proof Sketch]{%
  \begin{proof}[#1]%
}{%
  \end{proof}%
}

\newcommand{\Int}{\mathbf{Int}}
\newcommand{\Unit}{\mathbf{Unit}}

\newcommand{\typ}{\mathbf{typ}}
\newcommand{\lctx}{\mathbf{lctx}}
\newcommand{\gctx}{\mathbf{gctx}}
\newcommand{\sctx}{\mathbf{sctx}}
\newcommand{\heap}{\mathbf{heap}}
\newcommand{\spgv}[1]{\Sigma;\Psi;\Gamma\vdash_{#1}}
\newcommand{\sccv}[1]{\Sigma;\cdot;\cdot\vdash_{#1}}
\newcommand{\spccv}[1]{\Sigma';\cdot;\cdot\vdash_{#1}}
\newcommand{\depth}[2]{\mathrm{dep}_{#1}(#2)}
\newcommand{\size}[1]{\mathrm{siz}(#1)}
\newcommand{\idm}[1]{\iota(#1)}
\newcommand{\nb}[1]{\widehat{#1}}
\newcommand{\Tr}[1]{\mathrm{Tr}(#1)}
\newcommand\letin[3]{\mathbf{let}\ #1\leftarrow#2\ \mathbf{in}\ #3}
\newcommand{\N}{\mathbb{N}}
\newcommand{\dom}[1]{\mathrm{dom}(#1)}
\def\doublewedge{\mathrel{\wedge\!\!\!\wedge}}
\newcommand{\mergeConf} [2] {#1 \doublewedge #2}
\newcommand\cconf[2]{\mathsf{C}_{#1}^{#2}}
\newcommand\ass[2]{\mathit{assert}(#1\sim #2)}
\newcommand{\la}{\lambda}

\newcommand{\ra}{\rightarrow}

\newcommand{\xra}{\xrightarrow}
\newcommand{\xRa}{\xRightarrow}

\newcommand{\xrsa}[1]{\overset{#1}{\rightsquigarrow}}
\newcommand{\Da}{\Downarrow}
\newcommand{\da}{\downarrow}

\newcommand{\ol}{\overline}

\newcommand{\vd}{\vdash}
\newcommand{\ag}[1]{\langle{#1}\rangle}
\newcommand{\bc}[1]{\{#1\}}

\newcommand{\df}{\overset{\triangle}{=}}
\newcommand{\es}{{^\text{\textasciitilde}}}

\newcommand{\run}{\mathbf{run\ }}
\newcommand{\rf}{\mathbf{ref\ }}
\newcommand{\rec}[3]{\mathbf{rec}\ #1(#2).#3}
\newcommand{\bx}{\mathbf{box\ }}
\newcommand{\letbox}[3]{\mathbf{letbox}\ #1\leftarrow #2\ \mathbf{in}\ #3}
\newcommand{\FNames}{\mathrm{FNames}}

\newcommand{\BNames}{\mathrm{BNames}}

\newcommand{\Names}{\mathrm{Names}}
\newcommand{\FTyp}{\mathrm{FTyp}}
\newcommand{\BTyp}{\mathrm{BTyp}}
\newcommand{\AVal}[1]{\mathbf{AVal}(#1)}
\newcommand{\act}{\mathbf{a}}
\newcommand{\CC}{\mathbf{C}}
\newcommand{\tr}{\mathtt{t}}

\newcommand{\iif}[3]{\mathbf{if}\ #1\ \mathbf{then} \ #2\ \mathbf{else} \ #3}

\newcommand{\pfr}[1]{\mathit{pfr}_{#1}}
\newcommand{\pbr}[1]{\mathit{pbr}_{#1}}
\newcommand{\ofr}[1]{\mathit{ofr}_{#1}}
\newcommand{\obr}[1]{\mathit{obr}_{#1}}
\newcommand{\pqr}[1]{\mathit{pqr}_{#1}}

\newcommand{\rst}[2]{{#1}\upharpoonright_{#2}}

\newcommand{\lsm}{\lesssim}
\newcommand{\lsmc}{\lesssim^{\mathrm{CIU}}}

\newcommand\powerstaged{\mathit{power\_staged}}
\newcommand\powerstagedgen{\mathit{power\_staged\_gen}}
\newcommand\power{\mathit{power}}

\begin{document}

\maketitle

\begin{abstract}
MetaML-style metaprogramming languages allow programmers to construct, manipulate and run code.
In the presence of higher-order references for code,
ensuring type safety is challenging,
as free variables can escape their binders.
In this paper, we present Contextual MetaML,
\textit{the first metaprogramming language that supports storing and running open code
under a strong type safety guarantee}.
The type system utilises contextual modal types to track and reason about free variables in code explicitly.

A crucial concern in metaprogramming-based program optimisations
is whether the optimised program preserves the meaning of the original program.
Addressing this question requires a notion of program equivalence and techniques to reason about it.
In this paper, we provide a semantic model that captures contextual equivalence for Contextual MetaML,
establishing \textit{the first full abstraction result for an imperative MetaML-style language}.
Our model is based on traces derived via operational game semantics, 
where the meaning of a program is modelled by its possible interactions with the environment.
We also establish a novel closed instances of use theorem
that accounts for both call-by-value and call-by-name closing substitutions.
\end{abstract}

\section{Introduction}

\subsection{Type-safe metaprogramming with higher-order references for open code}

Generative metaprogramming languages enable programmers to write programs that generate programs.
The distinctive feature of quotation-based metaprogramming languages such as MetaML \cite{tahaMetaMLMultistageProgramming2000}
is their once-and-for-all type safety guarantee: well-typed metaprograms always generate well-typed object programs.
This stands in contrast to AST-based metaprogramming languages such as C++ Templates \cite{stroustrupProgrammingLanguage2013,vandevoordeTemplatesCompleteGuide2017}
and Template Haskell \cite{sheardTemplateMetaprogrammingHaskell2002,bergerModellingHomogeneousGenerative2017},
where the generated object program needs to undergo an additional typechecking phase.

MetaML's strong safety guarantee is difficult to achieve, especially in the presence of higher-order references
that can store code values.
For example, consider the following MetaOCaml \cite{kiselyovMetaOCamlTheoryImplementation2023,kiselyovDesignImplementationBER2014} program:
\begin{Ocaml}
let r = ref .< 1 >.;;
let f = .< fun x -> .~(r := .< x >.; .<0>.) >.;;
let t = Runcode.run (!r);;
\end{Ocaml}

The Bracket operator \ocaml{.< >.} converts a term into a static piece of code,
and the Escape operator \ocaml{.~} allows terms to be evaluated inside a Bracket.
In this example, the term \ocaml{r := .< x >.; .<0>.} is evaluated,
even if it occurs under the binder \ocaml{fun x -> }.
As a result, although the program appears to be well-scoped,
the reference \ocaml{r} stores the code \ocaml{.< x >.} which contains a free variable.
It is therefore unsafe to run the code stored in \ocaml{r}.

Currently, MetaOCaml raises a runtime exception
whenever there is an attempt to use a piece of code that contains free variables.
This solution is unsatisfactory for two reasons.
First, it breaks the type safety guarantee of MetaML,
as such scope extrusions are detected at runtime rather than at compile time.
Second, it limits the usability of the language,
as storing and manipulating code with free variables
is a useful tool in generative metaprogramming.

Recently, Hu and Pientka \cite{huLayeredModalType2024} developed a type theory for metaprogramming called
\textit{layered modal type theory}.
They use contextual types \cite{nanevskiContextualModalType2008}
to track and reason about free variables in code values explicitly,
and provide a strong safety guarantee on the generation and manipulation of such code values.
However, their result is based on an equational type theory that lacks general recursion and references.

In this paper, we present Contextual MetaML,
\textit{the first metaprogramming language that allows storing and running open code
under a strong type safety guarantee}.
The language has two key operators.
The Box operator, $\bx M$, converts any term $\Gamma\vd M:T$ into a piece of code.
The resulting type, $\Box(\Gamma\vd T)$, records both the type and the typing context of $M$.
The Letbox operator, $\letbox{u}{M_1}{M_2}$, ``unboxes'' the code $M_1$, and substitutes the extracted term for $u$ in $M_2$.
For such unboxing to be type-safe,
all occurrences of $u$ in $M_2$ must be associated with an explicit substitution $\delta$
that provides values for all the free variables in $M_1$.
This guarantees that however a piece of code is used or passed around in references,
its free variables can never be leaked in an undesirable context.

Consider the classical example of \textit{staging} the power function:
\begin{Ocaml}
let power n x =
    let y = ref 1 in
    for i = 1 to n do
        y := !y * x
    done;
    !y;;
\end{Ocaml}
We would like a function \ocaml{power_staged n}
whose output is a function that takes $x$ and computes $x^n$.
In MetaOCaml, one would have to write the following program:
\begin{Ocaml}
let rec power_staged_gen' n x =
    if n <= 0 then .< 1 >.
    else .< .~(power_staged_gen' (n - 1) x) * .~x >.;;
let power_staged' n =
    Runcode.run .< fun x -> .~(power_staged_gen' n .< x >.) >.;;
let f' = power_staged' 3;;
\end{Ocaml}
The function \ocaml{f'} evaluates to \ocaml{fun x -> 1 * x * x * x},
which is a more efficient program than the plain \ocaml{power n} function
as it no longer contains branching and recursion.
However, the programmer has to write an auxiliary recursive function \ocaml{power_staged_gen'},
and the pervasive use of Brackets and Escapes makes writing such programs
laborious and error-prone.

In Contextual MetaML,
one could write the program as\footnote{Here we use syntactic sugar,
omitting the explicit substitution $\delta$ if it is the identity, i.e. $[x/x]$.}:
\begin{Ocaml}
let power_staged n =
    let y = ref box 1 in
    for i = 1 to n do
        y := letbox u = !y in box (u * x)
    done;
    letbox u = !y in fun x -> u;;
let f = power_staged 3;;
\end{Ocaml}
Here, the code is constructed in a way
that is more intuitive and closer to the original imperative \ocaml{power} function.
The reference \ocaml{y} stores the open code \ocaml{box (1 * x * ...)} and is dynamically updated during the for-loop.

In \Cref{tab: type systems}, we summarise the design space of type systems for quotation-based metaprogramming.

\begin{table}
\begin{threeparttable}
\caption{Type systems for quotation-based metaprogramming languages}
\label{tab: type systems}
\setlength{\tabcolsep}{3pt}
\begin{tabular}{@{}cccccc@{}}
\toprule
Type system                     & Imperative & \makecell{Allows storing\\open code} & \makecell{Allows running\\open code} & \makecell{Type\\safety} & \makecell{Multi-\\layered} \\
\midrule
MetaML \cite{tahaMetaMLMultistageProgramming2000}                          & $\times$          & N/A                                        & $\times$                        & \checkmark           & \checkmark\\
Closed Types for MetaML \cite{calcagnoClosedTypesSafe2003}         & \checkmark          & $\times$                                          & $\times$                        & \checkmark          & \checkmark \\
Refined Environment Classifiers \cite{kiselyovRefinedEnvironmentClassifiers2016a} & \checkmark          & $\checkmark / \times$\tnote{1} & $\times$                        & \checkmark & $\times$          \\
MetaOCaml \cite{kiselyovMetaOCamlTenYears2024}                      & \checkmark          & \checkmark                                          & \checkmark                        & $\times$ & \checkmark          \\
Template Haskell \cite{sheardTemplateMetaprogrammingHaskell2002}               & $\times$          & N/A                                        & $\times$                        & $\checkmark / \times$ \tnote{2} & $\times$        \\
Typed Template Haskell \cite{xieStagingClassSpecification2022}         & $\times$          & N/A                                        & $\times$                        & \checkmark & $\times$          \\
MacoCaml \cite{xieMacoCamlStagingComposable2023}                      & \checkmark          & $\times$                                          & $\times$                        & \checkmark & $\times$          \\
M\oe bius \cite{jangMoebiusMetaprogrammingUsing2022}                        & $\times$          & N/A                                        & \checkmark                        & $\times$ & \checkmark          \\
Layered Modal Type Theory \cite{huLayeredModalType2024}      & $\times$          & N/A                                        & \checkmark                        & \checkmark  & $\times$         \\
$n$-layered Modal Type Theory \cite{huLayeredApproachIntensional2025}      & $\times$          & N/A                                        & \checkmark                        & \checkmark  & \checkmark         \\
Contextual MetaML (our work)            & \checkmark          & \checkmark                                          & \checkmark                        & \checkmark & $\times$          \\
\bottomrule
\end{tabular}
\begin{tablenotes}
\item [1] Refined Environment Classifiers allow storing open code temporarily, but such storage cannot go out of the range of the binder for the free variable.
\item [2] In Template Haskell, if a program is executed, it is guaranteed not to go wrong.
However, it requires typechecking any generated program before running it, so typechecking and program execution are interleaved.
\end{tablenotes}
\end{threeparttable}
\end{table}

\subsection{Closed instances of use with both call-by-name and call-by-value substitutions}
After introducing Contextual MetaML and proving its type safety,
we move on to discuss program equivalence in this language.
When performing metaprogramming-based program optimisations,
a critical concern is whether the optimisation is faithful.
In the above power example,
we expect that \ocaml{power} and \ocaml{power_staged} are equivalent.
To address this concern, we need to develop a theoretical foundation for the following question:
\textit{When are two Contextual MetaML programs equivalent?}

Capturing this problem denotationally is commonly known as the \textit{full abstraction} problem~\cite{milnerFullyAbstractModels1977},
where one seeks a sound and complete model for \textit{contextual equivalence}.
Two programs $M_1$ and $M_2$ are said to be contextually equivalent
if their observable behaviours (in our case, termination)
are the same in all suitable contexts.

The behaviour of a program in an arbitrary context is difficult to analyze for two reasons.
First, the program might be erased or duplicated before it is executed.
Second, the context can bind free variables in the program.
To overcome these difficulties,
programming language researchers utilise \textit{closed instances of use} (CIU)
\cite{honsellVariableTypedLogic1995,talcottReasoningProgramsEffects1998} approximations,
which say that equivalence under all suitable contexts
coincides with equivalence under all suitable evaluation contexts, heaps, and substitutions that close all free variables.

In Contextual MetaML, there are two kinds of variables.
Local variables $x,y,z$ are bound by $\lambda$-abstractions
and substituted in $\beta$-reductions.
Such variables are meant to represent values,
and are substituted in a call-by-value manner.
Global variables $u,v,w$ are bound by $\mathbf{letbox}$
and substituted in the reduction rule
\[(\letbox{u}{\bx M_1}{M_2},h)\ra (M_2[M_1/u],h)\]
Such variables are meant to represent unboxed code ($M_1$ in the above rule),
and are substituted in a call-by-name manner.
Therefore, our definition of closed instances of use
needs to take into account both kinds of substitutions.
For local variables, we consider all closing substitutions that map them to values,
while for global variables, we consider all closing substitutions that map them to arbitrary terms.
In short,
\textit{local variables represent closed values, global variables represent open terms\footnote{
    The word ``open'' here means that the term \textit{might} (but does not have to) contain free variables.
}.}

As an example, consider the following two terms with free local variable $x$
\[
\begin{aligned}
    x:\Int&\vd_0 x :\Int \\
    x:\Int&\vd_0 x;x :\Int \\
\end{aligned}
\]
where $M_1;M_2$ should be read as syntactic sugar for $(\la y.M_2)M_1$ for some fresh $y$.
All call-by-value substitutions will map $x$ to an integer value $\nb{n}$.
Under such substitutions, $\nb{n}$ and $\nb{n};\nb{n}$
are always equivalent as the latter $\beta$-reduces to the former in one step.
This is consistent with the fact that the two terms are not distinguishable in any context,
i.e., they are contextually equivalent.

On the other hand, consider the following two terms with free global variable $u$
\[
\begin{aligned}
    u:(\cdot\vd\Int)&\vd_0 u :\Int \\
    u:(\cdot\vd\Int)&\vd_0 u;u :\Int \\
\end{aligned}
\]
where the typing context $u:(\cdot\vd\Int)$ says that $u$ stands for a closed term of type $\Int$.
A call-by-name substitution may map $u$ to an arbitrary term $M$ of type $\Int$,
for example
\[\ell:=!\ell+1;!\ell\]
where $\ell$ is location of type $\rf Int$.
Under this substitution, the two terms are not equivalent.
Supposing the initial heap maps location $\ell$ to $0$,
then the first term evaluates to $1$ as $\ell$ is only updated once,
while the second term evaluates to $2$ as $\ell$ is updated twice.
Indeed, the two terms are not contextually equivalent as they can be distinguished by the context
\[\letin{x}{\rf 0}{\letbox{u}{\bx (x:=!x+1;!x)}{\bullet}}\]

\subsection{Operational game semantics for imperative metaprogramming}
Although CIU approximations suffice for small examples like above,
we need a more powerful tool to reason about contextual equivalences in general.
\textit{Operational game semantics} \cite{lairdFullyAbstractTrace2007,jaberCompleteTraceModels2021}
has been successfully used to construct fully abstract models for various programming languages.
It represents the behaviour of a program in a context
as a sequence of interactions between the program and the context.
An interaction is represented as an \textit{action},
and a sequence of actions is called a \textit{trace}.
The semantics of a program is then given by a set of
all the traces that the program can generate
according to a set of transition rules.
This approach is valued for its intuitive, operational nature
and is particularly well-suited to higher-order imperative programming languages.
Its prominent feature is the use of names to represent in abstract terms the functions
that are being called and communicated between the program and the context when they interact.

For example, consider a simplified trace
\[\ol{f_1}\quad f_1(3)\quad \ol{f_2}\quad f_2(2)\quad \ol{8}\]

Each action represents either a call or a return.
Actions performed by the program are overlined,
while those without an overline correspond to the context.
The program and the context alternate in taking actions as if it were a dialogue.
The symbols $f_1$ and $f_2$ are used in the dialogue between the Player and the Opponent
and are intended as abstract representations of function values,
in this case functions that are passed to the Opponent by the Player.
Since these are merely names, the Opponent gains no knowledge of the underlying function beyond its type.
In contrast, the Player will maintain a record mapping these names to the corresponding function values.

The trace above can be generated by the \ocaml{power_staged} program we mentioned earlier.
In this trace, the program first announces itself as a function represented by the name $f_1$ (return).
Then the context asks for the result of applying the function $f_1$ to $n=3$ (call),
and the program responds with another function name $f_2$ (return).
The name $f_2$ denotes the result of computing \ocaml{power_staged 3},
which is essentially the function \ocaml{fun x -> 1 * x * x * x},
but this underlying value is hidden from the context.
The context further supplies the argument $x=2$ to $f_2$ (call),
and the program responds that the result of the computation is $8$.
On the one hand, the trace can be generated from \ocaml{power_staged} alone according to our transition rules.
On the other hand, it also corresponds to the evaluation context $K=\bullet\ 3\ 2$.
This correspondence between traces and contexts is key to our full abstraction result,
and will be made precise in \Cref{thm:correctness} (Correctness) and \Cref{lem: definability} (Definability).

The trace above can also be generated by \ocaml{power}.
For \ocaml{power}, the underlying value of $f_2$ is
\begin{Ocaml}
fun x ->
    let y = ref 1 in
    for i = 1 to 3 do
        y := !y * x
    done;
    !y;;
\end{Ocaml}
Again, this underlying value is not reflected in the action,
so the difference from \ocaml{power_staged} is not observable to the context.

Writing $\Tr{M}$ for the set of traces corresponding to $M$ according to our model,
we shall have $\Tr{\text{\ocaml{power}}}=\Tr{\text{\ocaml{power_staged}}}$.
Given that the trace model is fully abstract, we can then deduce
$\text{\ocaml{power}}\approx\text{\ocaml{power_staged}}$
where $\approx$ stands for contextual equivalence.
This equivalence confirms the faithfulness of the staging transformation.

The trace above is not too different from
what we can find in existing trace models for higher-order functions and references \cite{lairdFullyAbstractTrace2007}.
The distinctive feature of our model is that the traces can also involve names that represent code values,
which we shall call \emph{box values}.
To illustrate this, let us consider the following program, reminiscent of the axiom $K$ for modal logic: $\Box (p\ra q)\ra \Box p\ra \Box q$.
\[
\begin{aligned}
\vd_0 &\la xy.\letbox{u}{x}\letbox{v}{y}\bx u^{x_2/x_1}\ v^{y_2/y_1} \\
\colon&\Box(x_1:\rf\Int\ra\Int\vd\rf\Int\ra \Int) \ra \Box(y_1:\rf\Int\vd\rf\Int) \\
&\ra \Box(x_2:\rf\Int\ra\Int,y_2:\rf\Int\vd\Int)
\end{aligned}
\]

The notation $V/x$ represents capture-avoiding substitution.
For example, the term $u^{x_2/x_1}$ stands for the result of replacing all occurrences of $x_1$ with $x_2$
in the unboxed term represented by $u$.

This program can generate the trace
\[
\begin{aligned}
&(\ol{f_1},\cdot,\cdot)\; (f_1(b_1),\cdot,\cdot)\; (\ol{f_2},\cdot,\cdot)\;
(f_2(b_2),\cdot,\cdot)\; (\ol{b_3},\cdot,\cdot)\; (\run b_3^{f_3/x_2,\ell_1/y_2},\Sigma,\chi_1)\\
&(\ol{\run b_1^{f_4/x_1}},\Sigma,\chi_1)\; (f_5,\Sigma,\chi_1)\;
(\ol{\run b_2^{\ell_1/y_1}},\Sigma,\chi_1)\; (\ell_1,\Sigma,\chi_1)\; (\ol{f_5}(\ell_1),\Sigma,\chi_1) \\
& (f_4(\ell_1),\Sigma,\chi_1)\; (\ol{f_3}(\ell_1),\Sigma,\chi_1)\; (1,\Sigma,\chi_2)\; (\ol{1},\Sigma,\chi_2)\; (1,\Sigma,\chi_2)\; (\ol{1},\Sigma,\chi_2)\\
\end{aligned}
\]
where $\Sigma=\ell_1:\rf \Int$ is a typing context for references,
and $\chi_1=\ell_1\mapsto 0$ and $\chi_2=\ell_1\mapsto 1$ are heaps.

Here, each action consists of three components.
The first component is the questions and answers that we discussed above.
The second and third components provide a heap $\Sigma;\cdot;\cdot\vd h:\heap$
that contains the references being shared by the program and the context.
In the trace, the program first announces itself as a function represented by the name $f_1$ (return).
Then the context asks for the result of applying $f_1$ to $x=b_1$ and $y=\ell_1$ sequentially (call),
triggering $\ol{f_2}$ and $\ol{b_3}$ (return).
The box name $b_3$ has type $\Box(x_2:\rf\Int\ra\Int,y_2:\rf\Int\vd\Int)$,
so in order to run it, the context needs to provide values for the variables $x_2$ and $y_2$,
which are of type $\rf\Int\ra\Int$ and $\rf\Int$ respectively.
In the corresponding action $\run b_3^{f_3/x_2,\ell_1/y_2}$ (call),
these values are represented by the function name $f_3$ and concrete location value $\ell_1$.
Now the program needs to evaluate $\#b_1^{f_3/x_1}\ \#b_2^{\ell_1/y_1}$,
where $\#b_i$ stands for the unboxed term extracted from box name $b_i$.
To evaluate this term, the program asks for the result of running
$b_1$ and $b_2$ (call) under corresponding substitutions.
Again, the function value $f_3$ is represented by a fresh name $f_4$,
while the location value $\ell_1$ is passed directly.
After obtaining the function name $f_5$ and location value $\ell_1$
as answers from the context (return),
the program can finally compute the result by applying $f_5$ to $\ell_1$ (call).
The context then asks for the result of applying $f_4$ to $\ell_1$ (call),
which the program can only answer after knowing the result of applying $f_3$ to $\ell_1$ (call).
The function $f_3$ modifies the value stored in location $\ell_1$ from $0$ to $1$
and returns the updated value $1$ (return).
Finally, a sequence of answers respond to all the pending questions so far (return).

The above trace corresponds to interacting with the evaluation context
\[
\begin{aligned}
&\letbox{u}{\bullet\ (\bx x_1)\ (\bx y_1)}\\
&\letin{x_3}{\la z.z:=!z+1;!z}\\
&\letin{y_3}{\rf 0}\\
&\quad u^{x_3/x_2,y_3/y_2}\\
\end{aligned}
\]

Our trace model is the first fully abstract model for an imperative MetaML-style metaprogramming language.
A comparison with similar results can be found in \Cref{tab: full abstraction}.

\begin{table}
\caption{Models of contextual equivalence for quotation-based metaprogramming languages}
\label{tab: full abstraction}
\begin{threeparttable}
\begin{tabularx}{\textwidth}{@{}>{\centering\arraybackslash}Xcccc@{}}
\toprule
Paper & Approach & Fully abstract & Imperative & Typed \\
\midrule
Taha and Sheard \cite{tahaMetaMLMultistageProgramming2000} & Equational rules & $\times$  & $\times$ & \checkmark \\
Berger and Tratt \cite{bergerProgramLogicsHomogeneous2015} & Hoare logic & \checkmark  & $\times$ & \checkmark \\
Inoue and Taha \cite{inoueReasoningMultistagePrograms2016} & Applicative bisimulation & \checkmark  & $\times$ & $\times$ \\
Our work & Operational game semantics & \checkmark  & \checkmark & \checkmark \\
\bottomrule
\end{tabularx}
\end{threeparttable}
\end{table}

\subparagraph*{Outline}
The paper is organized as follows.
In \Cref{sec: language}, we introduce the syntax, type system,
and operational semantics of Contextual MetaML,
and prove its type safety.
In \Cref{sec: ciu}, we define closed instances of use approximation for Contextual MetaML, and prove that it coincides with the contextual preorder.
In \Cref{sec: trace model}, we define the trace model for Contextual MetaML
based on operational game semantics.
In \Cref{sec: full abstraction}, we prove that the trace model is fully abstract with respect to contextual equivalence.
In \Cref{sec: examples}, we use examples to illustrate how our trace model
can be used to prove useful program equivalences in metaprogramming.
In \Cref{sec: related work,sec: future work},
we discuss related and future work respectively.

\section{Contextual MetaML}\label{sec: language}
In this section, we present the syntax, operational semantics, and type safety of Contextual MetaML.
\subsection{Syntax}
\begin{figure}
\begin{flushleft}

    \textbf{Types:} $T \df\Unit\mid\Int\mid T_1\ra T_2\mid\rf T\mid\Box(\Gamma\vd T)$
    
    \medskip

    \textbf{Store contexts:} $\Sigma\df\cdot\mid\Sigma,\ell:\rf{T}$

    \medskip

    \textbf{Local variable contexts:} $\Gamma \df\cdot\mid\Gamma,x:T$

    \medskip

    \textbf{Global variable contexts:} $\Psi\df\cdot\mid\Psi,u:(\Gamma\vd T)$

    \medskip

    \textbf{Terms: }
    \[\begin{array}{rcl}
   M &\df&()\mid \nb{n}\mid x\mid u^\delta\mid\ell
    \mid \la x.M\mid M_1M_2\mid \rec{y}{x}{M}
    \mid M_1\oplus M_2 \\
    &&\mid \iif{M_1}{M_2}{M_3}
    \mid \rf M\mid !M \mid M_1:=M_2 \mid \\
    &&\bx M\mid \letbox{u}{M_1}{M_2}
    \end{array}\]

    \textbf{Primitive values: } $a\,\df\, ()\mid\nb{n}\mid \ell$

    \medskip

    \textbf{Values: } $V\,\df\, a\mid x\mid \la x.M\mid \rec{y}{x}{M}\mid \bx M$

    \medskip

    \textbf{Local substitutions: } $\delta\df \cdot \mid\delta, V/x$

    \medskip

    \textbf{Global substitutions: } $\sigma\df \cdot \mid\sigma, M/u$

    \medskip

    \textbf{Heaps: } $h\df \cdot \mid h,\ell\mapsto V$

    \medskip

    \textbf{Contexts: }     
    \[\begin{array}{rcl}
    C&\df & \bullet \mid\la x.C
    \mid CM\mid MC \mid \rec{y}{x}{C} 
    \mid C\oplus M \mid M\oplus C \mid\\
    && \iif{C}{M_1}{M_2}\mid \iif{M_1}{C}{M_2}\mid\\
    &&\iif{M_1}{M_2}{C}\mid \rf C\mid !C\mid C:=M\mid M:=C\mid \\
    &&\bx C\mid \letbox{u}{C}{M}\mid \letbox{u}{M}{C}
    \end{array}\]

    \textbf{Evaluation contexts: }
    \[\begin{array}{rcl}
    K&\df & \bullet \mid KM\mid VK \mid K\oplus M \mid V\oplus K \mid
    \iif{K}{M_1}{M_2}\mid \\
    &&\rf{K}\mid !K\mid K:=M\mid V:=K\mid\letbox{u}{K}{M}\mid \\
    &&\letbox{u}{V}{K}
    \end{array}\]

\medskip

Throughout the Figure, we assume $x\in \mathit{LVar}$, $u\in\mathit{GVar}$, $\ell\in \mathit{Loc}$, $n\in\N$, and $\oplus\in\{+,-,*,<,=\}$.

\medskip

We write $\letin{x}{M_1}{M_2}$ for $(\lambda x.M_2)M_1$. If $x$ does not occur in $M_2$, we also write $M_1;M_2$.

\end{flushleft}
    \caption{Syntax}
    \label{fig: Syntax}
\end{figure}

The syntax of Contextual MetaML is given in \Cref{fig: Syntax}.
The language can be thought of as the result of adding layered modal type theory \cite{huLayeredModalType2024}
to ML with higher-order references.

There are two types of variables in Contextual MetaML.
Local variables, ranging over $x,y,z,\dots$,
are ordinary variables that can be found in any call-by-value lambda calculus.
An occurrence of $x$ in a term is either free or bound.
In the former case,
it shall appear in the local variable context $\Gamma$ as $x:T$,
which reads ``$x$ is of type $T$''.
In the latter case,
it shall be bound by a lambda abstraction $\la x.M$ or recursive definition $\rec{y}{x}{M}$.

Global variables, ranging over $u,v,w,\dots$,
are a special type of variables that are used in contextual type theories to represent open terms.
A global variable $u$ can never occur alone in a term;
it must always occur as $u^\delta$,
associated with an explicit local substitution $\delta$ that provides values for free variables in the term represented by $u$.
An occurrence of $u$ in a term is also either free or bound.
In the former case,
it shall appear in the global variable context $\Psi$ as $u:(\Gamma\vd T)$,
which reads ``$u$ is of type $T$ under local variable context $\Gamma$''.
In the latter case,
it shall be bound by a Letbox construct $\letbox{u}{M_1}{M_2}$,
which we shall explain now.

Contextual MetaML features two metaprogramming constructs.
The Box operator $\bx M$ converts a term $M$ into a value
that reads ``a piece of code whose content is $M$''.
The Letbox operator $\letbox{u}{M_1}{M_2}$
binds the occurrence of $u$ in $M_2$,
and reads ``In $M_2$, let $u$ be the term contained in $M_1$''.

To avoid confusion,
we shall use the word \textit{term} as defined in \Cref{fig: Syntax},
and the word \textit{code} to refer to a term of the form $\bx M$.
The word \textit{program} is reserved for use in the operational game semantics.

Readers familiar with the original MetaML \cite{tahaMetaMLMultistageProgramming2000}
might find it helpful to understand $\bx M$ as MetaML's Bracket $\langle M \rangle$,
and $\letbox{u}{M_1}{M_2}$ as MetaML's Escape $\es M$ and Run $\run M$.
Whether Letbox should be viewed as Escape or Run depends on whether $u$ occurs under a Box.
For example, $\letbox{u}{M}{u}$ corresponds to MetaML's $\run M$,
while $\letbox{u}{M}{\bx (1+u)}$ corresponds to MetaML's $\langle 1 + \es M \rangle$.
Compared with the original MetaML,
Contextual MetaML takes a more uniform and elegant approach to unboxing code,
and avoids the problem that \textit{evaluations can go under binders},
as noted by Inoue and Taha \cite{inoueReasoningMultistagePrograms2016}.
\subsection{Type System}
\begin{figure}

    \textbf{Well-formed types: } $\vd_i T:\typ$

    \begin{multicols}{3}
    \begin{prooftree}
        \AxiomC{$\vphantom{\vd_i T:\typ}$}
        \LeftLabel{(Typ-Unit)}
        \UnaryInfC{$\vd_i\Unit:\typ$}
    \end{prooftree}

    \begin{prooftree}
        \AxiomC{$\vphantom{\vd_i T:\typ}$}
        \LeftLabel{(Typ-Int)}
        \UnaryInfC{$\vd_i\Int:\typ$}
    \end{prooftree}

    \begin{prooftree}
        \AxiomC{$\vd_i T:\typ$}
        \LeftLabel{(Typ-Ref)}
    \UnaryInfC{$\vd_i \rf T:\typ$}
    \end{prooftree}

    \end{multicols}

    \begin{multicols}{2}

    \begin{prooftree}
        \AxiomC{$\vd_i T_1:\typ$}
        \AxiomC{$\vd_i T_2:\typ$}
        \LeftLabel{(Typ-Arrow)}
        \BinaryInfC{$\vd_i T_1\ra T_2:\typ$}
    \end{prooftree}

    \begin{prooftree}
        \AxiomC{$\vd_0 \Gamma:\lctx$}
        \AxiomC{$\vd_0 T:\typ$}
        \LeftLabel{(Typ-Box)}
        \BinaryInfC{$\vd_1 \Box(\Gamma\vd T):\typ$}
    \end{prooftree}

    \end{multicols}

    \textbf{Well-formed store contexts: } $\vd_i \Sigma:\sctx$
    \vspace{-3mm}
    \begin{multicols}{2}
    \begin{prooftree}
        \AxiomC{$\vphantom{\vd_i \Sigma:\sctx}$}
        \noLine
        \UnaryInfC{$\vphantom{\vd_i T:\typ \quad \ell\notin\dom{\Sigma}}$}
        \LeftLabel{(SCtx-Empty)}
        \UnaryInfC{$\vd_i \cdot:\sctx$}
    \end{prooftree}

    \begin{prooftree}
        \AxiomC{$\vd_i \Sigma:\sctx$}
        \noLine
        \UnaryInfC{$\vd_i T:\typ$}
        \AxiomC{$\ell\notin\dom{\Sigma}$}
        \LeftLabel{(SCtx-Cons)}
        \BinaryInfC{$\vd_i \Sigma,\ell:\rf T:\sctx$}
    \end{prooftree}
    \end{multicols}

    \textbf{Well-formed local variable contexts: } $\vd_i \Gamma:\lctx$
    \vspace{-3mm}
    \begin{multicols}{2}
    \begin{prooftree}
        \AxiomC{$\vphantom{\vd_i \Gamma:\lctx}$}
        \noLine
        \UnaryInfC{$\vphantom{\vd_i T:\typ \quad x\notin\dom{\Gamma}}$}
        \LeftLabel{(LCtx-Empty)}
        \UnaryInfC{$\vd_i \cdot:\lctx$}
    \end{prooftree}

    \begin{prooftree}
        \AxiomC{$\vd_i \Gamma:\lctx$}
        \noLine
        \UnaryInfC{$\vd_i T:\typ$}
        \AxiomC{$x\notin\dom{\Gamma}$}
        \LeftLabel{(LCtx-Cons)}
        \BinaryInfC{$\vd_i \Gamma,x:T:\lctx$}
    \end{prooftree}

    \end{multicols}

    \textbf{Well-formed global variable contexts: } $\vd \Psi:\gctx$
    \vspace{-3mm}
    \begin{multicols}{2}

    \begin{prooftree}
        \AxiomC{$\vphantom{\vd \Psi:\gctx \quad \vd_0 T:\typ}$}
        \noLine
        \UnaryInfC{$\vphantom{\vd_0 \Gamma:\lctx\quad u\notin\dom{\Psi}}$}
        \LeftLabel{(GCtx-Empty)}
        \UnaryInfC{$\vd \cdot:\gctx$}
    \end{prooftree}
    \begin{prooftree}
        \AxiomC{$\vd \Psi:\gctx$}
        \noLine
        \UnaryInfC{$\vd_0 \Gamma:\lctx$}
        \AxiomC{$\vd_0 T:\typ$}
        \noLine
        \UnaryInfC{$u\notin\dom{\Psi}$}
        \LeftLabel{(GCtx-Cons)}
        \BinaryInfC{$\vd \Psi,u:(\Gamma\vd T):\gctx$}
    \end{prooftree}

    \end{multicols}

    \begin{flushleft}

    Throughout the figure, we assume $i\in\{0,1\}$.
    
    \end{flushleft}

    \caption{Well-formed types and typing contexts}
    \label{fig: Type and Context}
\end{figure}

\begin{figure}

    \textbf{Well-typed terms: } $\spgv{i} M:T$
        
    \begin{prooftree}
        \AxiomC{$\vd_i\Gamma:\sctx$}
        \AxiomC{$\vd\Psi:\gctx$}
        \AxiomC{$\vd_i\Gamma:\lctx$}
        \LeftLabel{(Tm-Unit)}
        \TrinaryInfC{$\spgv{i}():\Unit$}
    \end{prooftree}

    \begin{prooftree}
        \AxiomC{$\vd_i\Gamma:\sctx$}
        \AxiomC{$\vd\Psi:\gctx$}
        \AxiomC{$\vd_i\Gamma:\lctx$}
        \LeftLabel{(Tm-Int)}
        \TrinaryInfC{$\spgv{i}\nb{n}:\Int$}
    \end{prooftree}

    \begin{prooftree}
        \AxiomC{$\vd_i\Sigma:\sctx$}
        \noLine
        \UnaryInfC{$\vd\Psi:\gctx$}
        \AxiomC{$\vd_i\Gamma:\lctx$}
        \noLine
        \UnaryInfC{$x:T\in\Gamma$}
        \LeftLabel{(Tm-LVar)}
        \BinaryInfC{$\spgv{i} x:T$}
    \end{prooftree}
    
    \begin{prooftree}
        \AxiomC{$\spgv{i}\delta:\Gamma'$}
        \AxiomC{$u:(\Gamma'\vd T)\in\Psi$}
        \LeftLabel{(Tm-GVar)}
        \BinaryInfC{$\spgv{i} u^\delta:T$}
    \end{prooftree}

    \begin{prooftree}
        \AxiomC{$\vd_i\Sigma:\sctx$}
        \AxiomC{$\vd\Psi:\gctx$}
        \AxiomC{$\vd_i\Gamma:\lctx$}
        \AxiomC{$\ell:\rf T\in\Sigma$}
        \LeftLabel{(Tm-Loc)}
        \QuaternaryInfC{$\spgv{i} \ell:\rf T$}
    \end{prooftree}

    \begin{prooftree}
        \AxiomC{$\Sigma;\Psi;\Gamma,x:T_1\vd_i M:T_2$}
        \LeftLabel{(Tm-Abs)}
        \UnaryInfC{$\spgv{i}\la x.M:T_1\ra T_2$}
    \end{prooftree}

    \begin{prooftree}
        \AxiomC{$\spgv{i} M_1:T_1\ra T_2$}
        \AxiomC{$\spgv{i} M_2:T_1$}
        \LeftLabel{(Tm-App)}
        \BinaryInfC{$\spgv{i} M_1M_2:T_2$}
    \end{prooftree}
    
    \begin{prooftree}
        \AxiomC{$\Sigma;\Psi;\Gamma,y:T_1\ra T_2,x:T_1\vd_i M:T_2$}
        \LeftLabel{(Tm-Rec)}
        \UnaryInfC{$\spgv{i} \rec{y}{x}{M}:T_1\ra T_2$}
    \end{prooftree}

    \begin{prooftree}
        \AxiomC{$\spgv{i} M_1:\Int$}
        \AxiomC{$\spgv{i} M_2:\Int$}
        \LeftLabel{(Tm-Arith)}
        \BinaryInfC{$\spgv{i} M_1\oplus M_2:\Int$}
    \end{prooftree}

    \begin{prooftree}
        \AxiomC{$\spgv{i} M_1:\Int$}
        \AxiomC{$\spgv{i} M_2:T$}
        \AxiomC{$\spgv{i} M_3:T$}
        \LeftLabel{(Tm-If)}
        \TrinaryInfC{$\spgv{i} \iif{M_1}{M_2}{M_3}:T$}
    \end{prooftree}

    \begin{multicols}{2}
    \begin{prooftree}
        \AxiomC{$\spgv{i} M: T$}
        \LeftLabel{(Tm-Ref)}
        \UnaryInfC{$\spgv{i} \rf M:\rf T$}
    \end{prooftree}

    \begin{prooftree}
        \AxiomC{$\spgv{i} M:\rf T$}
        \LeftLabel{(Tm-Deref)}
        \UnaryInfC{$\spgv{i} !M:T$}
    \end{prooftree}

    \end{multicols}

    \begin{prooftree}
        \AxiomC{$\spgv{i} M_1:\rf T$}
        \AxiomC{$\spgv{i} M_2:T$}
        \LeftLabel{(Tm-Assign)}
        \BinaryInfC{$\spgv{i} M_1:=M_2:\Unit$}
    \end{prooftree}

    \begin{prooftree}
        \AxiomC{$\vd_1\Gamma:\lctx$}
        \AxiomC{$\Sigma;\Psi;\Gamma'\vd_0 M:T$}
        \LeftLabel{(Tm-Box)}
        \BinaryInfC{$\spgv{1} \bx M:\Box(\Gamma'\vd T)$}
    \end{prooftree}

    \begin{prooftree}
        \AxiomC{$\spgv{1} M_1:\Box(\Gamma'\vd T)$}
        \AxiomC{$\Sigma;\Psi,u:(\Gamma'\vd T);\Gamma\vd_1 M_2:T'$}
        \LeftLabel{(Tm-LetBox)}
        \BinaryInfC{$\spgv{1} \letbox{u}{M_1}{M_2}:T'$}
    \end{prooftree}
    
    \medskip

    \begin{flushleft}

    Throughout the figure, we assume $i\in\{0,1\}$.
    
    \end{flushleft}

    \caption{Well-typed terms}
    \label{fig: Type System}
\end{figure}

\begin{figure}
    \textbf{Well-typed heaps:} $\spgv{} h:\heap$
    \begin{prooftree}
        \AxiomC{$\dom{\Sigma}=\dom{h}$}
        \AxiomC{$\spgv{i}h(\ell):T$ for all $\ell:\rf T\in\Sigma$}
        \LeftLabel{(Heap)}
        \BinaryInfC{$\spgv{}h:\heap$}
    \end{prooftree}

    \textbf{Well-typed contexts:} $\sccv{} C:(\Psi';\Gamma';i';T')\ra(\Psi;\Gamma;i;T)$
    \begin{prooftree}
        \AxiomC{$\Psi'\subseteq\Psi$}
        \AxiomC{$\Gamma'\subseteq\Gamma$}
        \AxiomC{$i'\le i$}
        \LeftLabel{(Ctx-Hole)}
        \TrinaryInfC{$\sccv{} \bullet:(\Psi';\Gamma';i';T)\ra(\Psi;\Gamma;i;T)$}
    \end{prooftree}
    Other rules are omitted for brevity.

    \textbf{Well-typed local substitutions:} $\spgv{i}\delta:\Gamma'$

    \begin{prooftree}
        \AxiomC{$\dom{\delta}=\dom{\Gamma'}$}
        \AxiomC{$\spgv{i} \delta(x):T$ for all $x:T\in\Gamma'$}
        \LeftLabel{(LSubst)}
        \BinaryInfC{$\spgv{i}\delta:\Gamma'$}
    \end{prooftree}

    \textbf{Well-typed global substitutions:} $\Sigma;\Psi\vd \sigma:\Psi'$

    \begin{prooftree}
        \AxiomC{$\dom{\sigma}=\dom{\Psi'}$}
        \AxiomC{$\Sigma;\Psi;\Gamma\vd_0 \sigma(u):T$ for all $u:(\Gamma\vd T)\in\Psi'$}
        \LeftLabel{(GSubst)}
        \BinaryInfC{$\Sigma;\Psi\vd \sigma:\Psi'$}
    \end{prooftree}

    \medskip

    \begin{flushleft}

    Throughout the Figure, we assume $i\in\{0,1\}$.
    
    \end{flushleft}

    \caption{Well-typed heaps, contexts and substitutions}
    \label{fig: Heap Context and Substitution}
\end{figure}

The typing rules of Contextual MetaML are presented in
\Cref{fig: Type and Context,fig: Type System,fig: Heap Context and Substitution}.

The way global variables are typed in the rule (Tm-GVar) should be read as follows:
if $u$ represents a term of type $T$ under local variable context $\Gamma'$,
and $\delta$ is a type-preserving local substitution that maps variables in $\Gamma'$
to values under the local variable context $\Gamma$,
then $u^\delta$ is a term of type $T$ under $\Gamma$.

The new Box type $\Box(\Gamma\vd T)$ shall be read as
``a piece of code whose content is a term of type $T$ under local variable context $\Gamma$''.
The rules (Tm-Box) and (Tm-Letbox) are the introduction and elimination rules for this type respectively.

In the rule (Tm-Box),
the $\bx M$ operator binds all the free local variables in $M$.
The typing context required to type $M$, i.e. $\Gamma'$,
is encoded in the type of $\bx M$, i.e., $\Box(\Gamma'\vd T)$.
The resulting term $\bx M$ can therefore be typed under any, possibly empty,
local variable context $\Gamma$.
For example, we have $\vd_1 \bx x :\Box(x:\Int\vd \Int)$.

The rule (Tm-Letbox)
links two usages of the contextual type $\Gamma\vd T$ in the type system:
in a global variable context as $u:(\Gamma\vd T)$,
and in a Box type as $\Box(\Gamma\vd T)$.
The rule says that if $u$ is used as a term of type $T$ under $\Gamma'$ in $M_2$,
then the code that it is substituted for, i.e. $M_1$,
must indeed be a piece of code whose content has type $T$ under $\Gamma'$.

Most of the typing judgments are annotated with a layer $i$
that tracks the maximum depth of nested Boxes in all the associated types.
For example, we have $x:\Box(\cdot\vd\Int)\vd_1 x:\Box(\cdot\vd\Int)$,
and this judgment is not valid for layer $i=0$,
even if the term $x$ itself does not contain any Box.
In this paper, we follow Hu and Pientka \cite{huLayeredModalType2024}
and restrict the language to contain two layers only.
In other words, we forbid nesting of Box types at all.
This restriction greatly simplifies the metatheory of the language
while still allowing most of the common usages of metaprogramming
as we shall show in various examples in \Cref{sec: examples}.
We discuss the possible extension to multiple layers in \Cref{sec: future work}.

\subsection{Operational Semantics}

\begin{figure}
    \textbf{Reduction rules:} $(M,h)\ra(M',h')$
    \begin{multicols}{2}
    
    \begin{prooftree}
        \AxiomC{}
        \LeftLabel{($\beta$)}
        \UnaryInfC{$((\la x.M)V,h)\ra (M[V/x],h)$}
    \end{prooftree}

    \begin{prooftree}
        \AxiomC{}
        \LeftLabel{(Arith)}
        \UnaryInfC{$(\nb{n_1}\oplus\nb{n_2},h)\ra (\nb{n_1\oplus n_2},h)$}
    \end{prooftree}

    \end{multicols}

    \begin{prooftree}
        \AxiomC{$n\neq 0$}
        \LeftLabel{(If true)}
        \UnaryInfC{$(\iif{\nb{n}}{M_1}{M_2},h)\ra (M_1,h)$}
    \end{prooftree}

    \begin{prooftree}
        \AxiomC{}
        \LeftLabel{(If false)}
        \UnaryInfC{$(\iif{\nb{0}}{M_1}{M_2},h)\ra (M_2,h)$}
    \end{prooftree}

    \begin{multicols}{2}

    \begin{prooftree}
        \AxiomC{$\ell\notin\dom{h}$}
        \LeftLabel{(Ref)}
        \UnaryInfC{$(\rf V,h)\ra (\ell,h\uplus[\ell\mapsto V])$}
    \end{prooftree}

    \begin{prooftree}
        \AxiomC{$\ell\in\dom{h}$}
        \LeftLabel{(Deref)}
        \UnaryInfC{$(!\ell,h)\ra (h(\ell),h)$}
    \end{prooftree}

    \end{multicols}

    \begin{prooftree}
        \AxiomC{$\ell\in\dom{h}$}
        \LeftLabel{(Assign)}
        \UnaryInfC{$(\ell:=V,h)\ra ((),h[\ell\mapsto V])$}
    \end{prooftree}

    \begin{prooftree}
        \AxiomC{}
        \LeftLabel{(Rec)}
        \UnaryInfC{$((\rec{y}{x}{M})V,h)\ra (M[V/x,\rec{y}{x}{M}/y],h)$}
    \end{prooftree}

    \begin{prooftree}
        \AxiomC{}
        \LeftLabel{(Letbox)}
        \UnaryInfC{$(\letbox{u}{\bx M_1}{M_2},h)\ra (M_2[M_1/u],h)$}
    \end{prooftree}

    \begin{prooftree}
        \AxiomC{$(M_1,h_1)\ra(M_2,h_2)$}
        \LeftLabel{(Ktx)}
        \UnaryInfC{$(K[M_1],h_1)\ra(K[M_2],h_2)$}
    \end{prooftree}
    \begin{flushleft}
        We write $\ra^*$ for the reflexive transitive closure of $\ra$.
        We write $(M,h)\da$ if there exists value $V$ and heap $h'$ such that $(M,h)\ra^*(V,h')$.
    \end{flushleft}

    \medskip

    \textbf{Applying local substitutions: } $M[\delta]$
    \vspace{-2mm}
    \begin{multicols}{2}
        \noindent
        \[
        \begin{array}{rcl}
        ()[\delta] &=& ()\\
        (\nb{n})[\delta] &=& \nb{n}\\
        x[\delta] &=& \delta(x)
        \end{array}
        \]
        \[
        \begin{array}{rcl}
        u^\delta[\delta'] &=& u^{\delta[\delta']}\\
        \ell[\delta] &=& \ell\\
        (\la x.M)[\delta] &=& \la x.M[\delta,x/x]
        \end{array}
        \]
    \end{multicols}
    \vspace{-8mm}
    \[
    \begin{array}{rcl}
        (\rec{y}{x}{M})[\delta] &=& \rec{y}{x}{M[\delta,y/y,x/x]} \\
        (\bx M)[\delta] &=& \bx M \\
        (\letbox{u}{M_1}{M_2})[\delta] &=& \letbox{u}{M_1[\delta]}{M_2[\delta]}
    \end{array}
    \]
    \begin{flushleft}
        where $\delta[\delta'](x)=\delta(x)[\delta']$, and $[\delta]$ commutes with all other term constructors.
    \end{flushleft}

    \medskip
    \textbf{Applying global substitutions: } $M[\sigma]$
    \vspace{-2mm}
    \begin{multicols}{2}
        \noindent
        \[
        \begin{array}{rcl}
        ()[\sigma] &=& ()\\
        (\nb{n})[\sigma] &=& \nb{n} \\
        x[\sigma] &=& x
        \end{array}
        \]
        \[
        \begin{array}{rcl}
        u^\delta[\sigma] &=& \sigma(u)[\delta[\sigma]] \\
        \ell[\sigma] &=& \ell \\
        (\la x.M)[\sigma] &=& \la x.M[\sigma]
        \end{array}
        \]
    \end{multicols}
    \vspace{-8mm}
    \[
    \begin{array}{rcl}
        (\rec{y}{x}{M})[\sigma] &=& \rec{y}{x}{M[\sigma]} \\
        (\bx M)[\sigma] &=& \bx (M[\sigma]) \\
        (\letbox{u}{M_1}{M_2})[\sigma] &=& \letbox{u}{M_1[\sigma]}{M_2[\sigma,u/u]}
    \end{array}
    \]
    \begin{flushleft}
        where $\delta[\sigma](x)=\delta(x)[\sigma]$, and $[\sigma]$ commutes with all other term constructors.
    \end{flushleft}
\caption{Operational Semantics}
\label{fig: Operational Semantics}
\end{figure}

The small-step operational semantics of Contextual MetaML is presented in \Cref{fig: Operational Semantics}.
The rule (Letbox) does two things simultaneously.
First, it extracts the content $M_1$ from the code $\bx M_1$, reminiscent of MetaML's Escape and Run.
Second, it substitutes $M_1$ for all the occurrences of the global variable $u$ in $M_2$,
reminiscent of call-by-name $\beta$-reductions.

Applying the substitutions created by the $(\beta)$, (Rec), and (Letbox) rules needs extra care.
$\la x.M$ and $\rec{y}{x}{M}$ bind the mentioned local variables,
and $\bx M$ binds all local variables.
Similarly, $\letbox{u}{M_1}{M_2}$ binds the global variable $u$ in $M_2$.
When a substitution goes under these binders,
the bound variables are removed from the substitution, and all other variables are preserved.

When we write a term $M[\delta]$ or $M[\sigma]$,
the mentioned substitution is to be viewed as an implicit substitution.
It is performed immediately,
and is gone in the resulting term.
When we write a term $u^\delta$,
the mentioned substitution is to be viewed as an explicit substitution,
and the term stays as it is.
When applying an implicit substitution $\sigma$ or $\delta'$ to the term $u^\delta$,
the two substitutions are composed,
yielding a new explicit substitution $\sigma[\delta]$ or $\delta'[\delta]$.

\begin{remark}
    One might have the intuition that a box can be mimicked by a function taking a unit argument,
    and unboxing can be mimicked by applying such a function to the unit value.
    Although this intuition can be helpful in establishing \Cref{lem: CIU box} below,
    there does not exist a translation from Contextual MetaML to ML.
    Consider the term $\la x.\letbox{u}{x}{\bx(u^{y/y}+y)}$.
    It takes as input some code containing free variable $y$,
    and produces the code with ``$+y$'' appended to it.
    For example, it transforms $\bx (y*y)$ to $\bx (y*y+y)$.
    There is no corresponding ML function that can transform the function $\la z.y*y$ to $\la z.y*y+y$.
    Similarly, there is no ML counterpart to $\la x. \letbox{u}{x}{\bx(u + 1)}$
    that takes a closed function and appends “+1” to it.
\end{remark}

\subsection{Type Safety}
\begin{theorem}[Preservation]\label{thm: preservation}
    If $\spgv{i} M:T$, $\spgv{} h:\heap$, and $(M,h)\ra(M',h')$,
    then there exists $\Sigma'\supseteq\Sigma$ such that
    $\Sigma';\Psi;\Gamma\vd_i M':T$, and $\Sigma';\Psi;\Gamma\vd h':\heap$.
\end{theorem}

\begin{lemma}[Canonical forms]\label{lem: canonical forms}
    If $\sccv{i} V:T$, then one of the following cases holds:
    \begin{itemize}
        \item $T=\Unit$ and $V=()$.
        \item $T=\Int$ and $V=\nb{n}$ for some $n\in\N$.
        \item $T=\rf T'$ and $V=\ell$ for some $\ell\in\mathit{Loc}$.
        \item $T=T_1\ra T_2$ and $V=\la x.M$ or $V=\rec{y}{x}{M}$.
        \item $T=\Box(\Gamma\vd T')$ and $V=\bx M$.
    \end{itemize}
\end{lemma}

\begin{theorem}[Progress]\label{thm: progress}
    If $\sccv{i} M:T$ and $\sccv{} h:\heap$,
    then either $M$ is a value, or there exist $M'$, $h'$ such that
    $(M,h)\ra(M',h')$.
\end{theorem}

In the rest of this paper, we shall consider 
the following notion of equivalence for Contextual MetaML.

\begin{definition}[Contextual preorder]\label{def: contextual preorder}
    Given two terms $\spgv{i} M_1, M_2:T$,
    we define $M_1\lsm M_2$ to hold if,
    for any context $\spccv{} C: (\Psi;\Gamma;i;T)\ra(\cdot;\cdot;i';T')$
    and heap $\spccv{} h:\heap$ such that $\Sigma'\supseteq\Sigma$,
    if $(C[M_1],h)\Da$ then $(C[M_2],h)\Da$.
\end{definition}

We write $M_1\approx M_2$ if $M_1\lsm M_2$ and $M_2\lsm M_1$.

\section{Closed Instances of Use}\label{sec: ciu}

Closed Instances of Use approximation \cite{talcottReasoningProgramsEffects1998}
is widely used in reasoning about program equivalence in higher-order call-by-value languages.
It says that two terms are contextually equivalent
if and only if they are equivalent under all suitable evaluation contexts,
heaps, and closing substitutions.

Contextual MetaML has two kinds of variables.
Local variables are similar to normal variables in call-by-value languages,
so for every free variable declared as $x:T$ in the local variable context,
the closing substitution shall substitute a closed value of type $T$.

Global variables, on the other hand, are declared as $u:\Gamma\vd T$ in the global variable context,
which says that $u$ represents a term of type $T$ under local variable context $\Gamma$.
Also, when a substitution for $u$ is created by the (Letbox) rule,
what gets substituted is exactly a term of type $T$ under local variable context $\Gamma$.
Therefore, for every such global variable $u$,
the closing substitution shall substitute a term of type $T$ under local variable context $\Gamma$.

The distinction between local and global variables can now be summarised as:
``local variables represent closed values, global variables represent open terms.''
Note that the word ``open'' here means that the term might contain free variables,
but does not have to.

\begin{definition}[Closed Instances of Use]\label{def: closed instances of use}
    Given two terms $\spgv{i} M_1, M_2:T$,
    we define $M_1\lsmc M_2$ to hold if,
    for any global substitution $\Sigma';\cdot\vd \sigma:\Psi$,
    local substitution $\spccv{i} \delta:\Gamma$,
    evaluation context $\spccv{} K: (\cdot;\cdot;i;T)\ra(\cdot;\cdot;i';T')$,
    and heap $\spccv{} h:\heap$
    such that $\Sigma'\supseteq\Sigma$,
    if $(K[M_1[\sigma][\delta]],h)\Da$ then $(K[M_2[\sigma][\delta]],h)\Da$.
\end{definition}

Note that although we write the global substitution as $\Sigma';\cdot\vd \sigma:\Psi$,
we are not saying that $\sigma(u)$ must be a closed term.
By definition, this notation means $\Sigma';\cdot;\Gamma\vd \sigma(u):T$
for all $u:\Gamma\vd T$ in $\Psi$.
In other words, $\sigma(u)$ can contain free local variables as declared in $\Psi$,
but should not introduce any new free global variables.

\begin{lemma}[CIU Basics]\label{lem: CIU basics}
    Given $\spgv{i} M_1,M_2:T$, if $M_1\lsmc M_2$,
    then $C[M_1]\lsmc C[M_2]$ whenever the terms are typable and $C$ is one of the following forms:
    \ $\bullet M$, \ $\bullet \oplus M$, \ $\iif{\bullet}{M}{M'}$, \
    $\rf \bullet$, \ $!\bullet$, \ $\bullet := M$, \
    $\letbox{u}{\bullet}{M}$, \ $M\bullet$, \ $M\oplus\bullet$, \
    $\iif{M}{\bullet}{M'}$, \ $\iif{M}{M'}{\bullet}$, \
    $\la x.\bullet$, \ $\rec{y}{x}{\bullet}$, \ $\letbox{u}{M}{\bullet}$.
\end{lemma}

\begin{lemma}[CIU Box]\label{lem: CIU box}
    Given $\spgv{0} M_1,M_2:T$, if $M_1\lsmc M_2$,
    then $\bx M_1\lsmc \bx M_2$.
\end{lemma}
\begin{proof_sketch}
The ``elimination rule'' for $\bx M_i$, i.e.,
$(\letbox{u}{\bx M_i}{M},h) \ra (M[M_i/u],h)$,
ends up in a term where $M_i$ does not appear in an evaluation context,
so we cannot use the assumption that $M_1\lsmc M_2$ directly.
To solve this problem,
we utilise the (imperfect) intuition that a function taking unit argument mimics a box,
and applying such function to unit mimics unboxing.
For each occurrence of $M_i$ in the resulting term,
we can convert the term to an equivalent form,
where $\la y.M_i$ appears in an evaluation context.
Then we can use \Cref{lem: CIU basics} to establish the relation.
\end{proof_sketch}

\begin{theorem}[CIU]\label{thm: CIU}
    Given two terms $\spgv{i} M_1, M_2:T$, we have $M_1\lsm M_2$ if and only if $M_1\lsmc M_2$.
\end{theorem}
\section{Trace Model}\label{sec: trace model}

This section presents the trace model of Contextual MetaML. 
To this end, we shall introduce a labelled transition system
that will generate traces for a given term.
Intuitively, they can be viewed as an abstract record of interaction between the term and a context.
In particular, whenever functions or pieces of code are exchanged between the term and the context,
the values will be represented symbolically using names,
so that syntactically different but semantically equivalent values will not be distinguishable.
In a similar spirit, a location in the program heap will be visible to the context only if it has been revealed.
\subsection{Names}

Names are a representation of closed higher-order values.
They expose the type of a value but encapsulate the underlying syntax.
We will rely on an infinite set of names $\Names=\FNames\uplus\BNames$,
partitioned into \emph{function names} ($\FNames$) and
\emph{box names} ($\BNames$).
The two sets will be used to represent function and box (code) values respectively.
In line with \Cref{lem: canonical forms},
we assume that $\FNames$ and $\BNames$
are further partitioned by
$\FNames=\biguplus_{(T_1,T_2,i)\in \FTyp}\FNames_{T_1\ra T_2,i}$ where $\FTyp=\bc{(T_1,T_2,i)\mid\ \vd_i T_1\ra T_2:\typ}$,
and $\BNames=\biguplus_{(\Gamma,T)\in \BTyp}\BNames_{\Box(\Gamma\vd T),1}$ where $\BTyp=\bc{(\Gamma,T)\mid\ \vd_1 \Box(\Gamma\vd T):\typ}$,
to reflect the possible types of values.
We assume that all the constituent sets are countably infinite.
We will use $f$, $b$ and $n$ to range over $\FNames$, $\BNames$ and $\Names$ respectively.
We have typing rule
\begin{prooftree}
    \AxiomC{$n\in \Names_{T,i}$}
    \AxiomC{$i' \ge i$}
    \LeftLabel{(Tm-Name)}
    \BinaryInfC{$\spgv{i'} n:T$}
\end{prooftree}

A name $n$ is a closed value in the operational semantics,
and subsumes any substitutions, i.e.,
$n[\delta]=n[\sigma]=n$.

To handle box names operationally, we add one special reduction rule
\begin{prooftree}
    \AxiomC{}
    \LeftLabel{(Letbox name)}
    \UnaryInfC{$(\letbox{u}{b}{M},h)\ra (M[\#b/u],h)$}
\end{prooftree}
where the symbol $\# b$ denotes the result of removing the outermost box of $b$.
Like $u$, all the occurrences of $\#b$ should be associated with a local substitution
that closes all the free local variables in it.
Accordingly, we have typing rule
\begin{prooftree}
    \AxiomC{$b\in\Names_{\Box(\Gamma'\vd T),1}$}
    \AxiomC{$\spgv{i} \delta:\Gamma'$}
    \LeftLabel{(Tm-Unbox)}
    \BinaryInfC{$\spgv{i} \#b^\delta:T$}
\end{prooftree}
Substitutions on $\#b^\delta$ are defined as
$(\#b^\delta)[\delta'] = \#b^{\delta[\delta']}$
and $(\#b^\delta)[\sigma] = \#b^{\delta[\sigma]}$.
The term $\#b^\delta$ is a redex in the operational semantics.

For a unified treatment of values of box type,
we define $\#(\bx M)=M$.
One can verify that the above statements also hold if we change all occurrences of $b$ to $\bx M$.
\subsection{Value abstraction}
As already mentioned, function and box values will be abstracted by names in our traces,
while ground values and locations will be represented literally.
To keep track of abstractions,
we need a notation that links values to the set of their possible abstractions.
When the abstraction is a name,
we will also need to record the correspondence between the name and the underlying value. 

\begin{definition}[value abstraction]\label{def: abstract values}
    Given a closed value $\sccv{i} V:T$,
    we define its \textit{abstraction set} $\AVal{\sccv{i} V:T}$
    by cases from \Cref{lem: canonical forms}.
    \begin{itemize}
        \item $\AVal{\sccv{i} ():\Unit}\df \bc{((),\emptyset)}$.
        \item $\AVal{\sccv{i} \nb{n}:\Int}\df \bc{(\nb{n},\emptyset)}$.
        \item $\AVal{\sccv{i} \ell:\rf T}\df \bc{(\ell,\emptyset)}$.
        \item $\AVal{\sccv{i} V: T_1\ra T_2}\df \bc{(f,f\mapsto V)\mid f\in \FNames_{T_1\ra T_2,i}}$.
        \item $\AVal{\sccv{1} V: \Box(\Gamma\vd T)}\df \bc{(b,b\mapsto V)\mid b\in \BNames_{\Box(\Gamma\vd T),1}}$.
    \end{itemize}
\end{definition}

We use $A$ to refer to any abstraction,
i.e. $A$ ranges over unit, numerals, locations and names.
We use $\nu(A)$ to denote the set of names occurring in $A$,
and $\mu(A)$ for the set of locations occurring in $A$.
Both $\nu(A)$ and $\mu(A)$ are either empty or singleton.

Having defined what it means to abstract values,
we now lift the process to local substitutions and heaps,
where we want to replace all the participating values with corresponding abstractions.
We use $\rho$ and $\chi$ to refer to abstract local substitutions and heaps respectively.

Given a local substitution $\sccv{i}\delta:\Gamma$,
we let $\AVal{\sccv{i}\delta:\Gamma}\df
\{(\rho,\bigcup_x\gamma_x)\mid
\dom{\rho}=\dom{\delta},
(\rho(x),\gamma_x)\in\AVal{\sccv{i} \delta(x):\Gamma(x)},
\cap_{x\in\dom{\rho}}\nu(\rho(x))=\emptyset\}$.
Given a heap $\sccv{} h:\heap$, 
we let $\AVal{\sccv{} h:\heap}\df
\{(\chi,\bigcup_{\ell}\gamma_\ell)\mid
\dom{\chi}=\dom{h},
(\chi(\ell),\gamma_\ell)\in\AVal{\sccv{i} h(\ell):\Sigma(\ell)},
\cap_{\ell\in\dom{\chi}}\nu(\chi(\ell))=\emptyset\}$.
We define $\nu(\chi)=\biguplus_{\ell}\nu(\chi(\ell))$, i.e. $\nu(\chi)$ contains all names in $\chi$.
We define $\mu(\rho)=\biguplus_{x}\mu(\rho(x))$, i.e. $\mu(\rho)$ contains all locations in $\rho$.

\subsection{Actions}
In our model, the meaning of a program will be interpreted as a set of traces.
A trace will be a sequence of actions alternating between the program
and the context. Traditionally in game semantics, the program is referred to as P (Player, Proponent)
and the context as O (Opponent).
We will stick to this terminology below for compatibility with earlier works in the area.
Intuitively, the actions should be viewed as abstractions of calls and returns made by the program or context.
Traditionally, the former are called questions and the latter are answers.
Note that actions will also feature abstract heaps $\Sigma;\cdot;\cdot\vd\chi:\heap$,
which are an abstraction of the part 
of the heap that was exposed during the interaction. 

Our semantics will be based on the following kinds of actions.
\begin{itemize}
    \item Player Answer (PA) $(\ol{A},\Sigma,\chi)$.
    Program replies to a previous question with the answer $\sccv{i} A:T$,
    setting the values of shared locations to $\sccv{}\chi:\heap$.
    \item Player Question of Function (PQF) $(\ol{f}(A),\Sigma,\chi)$.
    Program asks for the result of applying $f\in \FNames_{T_1\ra T_2, i}$ to the abstract value $\sccv{i} A:T_1$,
    setting the values of shared locations to $\sccv{}\chi:\heap$.
    \item Player Question of Box (PRB) $(\ol{\run b^\delta},\Sigma,\chi)$.
    Program asks for the result of running (unboxing) the box name $b\in\BNames_{\Box(\Gamma\vd T)}$
    under abstract local substitution $\sccv{i}\rho:\Gamma$,
    setting the values of shared locations to $\sccv{}\chi:\heap$.
    \item Opponent Answer (OA) $(A,\Sigma,\chi)$, Opponent Question of Function (OQF) $(f(A),\Sigma,\chi)$,
    and Opponent Question of Box (OQB) $(\run b^\rho,\Sigma,\chi)$ are similar,
    except that they are initiated by the context.
\end{itemize}

An action is legitimate only if
$\chi$ is closed with respect to reachable locations and
disjoint names are used for abstraction:
$\chi(\ell)=\ell'$ implies $\ell'\in\dom{\chi}$ and $\bigcap_{\ell\in\dom{\chi}}\nu(\chi(\ell))=\emptyset$.

\subsection{Configurations and transitions}

Finally, we define the labelled transition system (LTS) that will generate traces from terms.
We start off by discussing the shape of its configurations.

A \emph{passive configuration} describes a moment when the context (Opponent) is about to initiate a computation.
It is a tuple $(\Sigma, h,\gamma,\phi,\xi,\theta)$ where
$\sccv{} h:\heap$ is the heap, $\gamma$ maps function and box names to the value that they represent,
$\phi\subseteq\Names$ collects all the names that have been introduced so far,
$\xi:=\cdot\mid \xi:(\Sigma_1\vd K:(\cdot;\cdot;T_1;i_1)\ra(\cdot;\cdot;T_2;i_2))$ is a stack of evaluation contexts,
and $\theta\subseteq \dom{h}$ is the set of locations currently shared by program and context.

Following an action by the context,
passive configurations will evolve into active ones.
\emph{Active configurations} have the form
$(M,i,T,\Sigma,h,\gamma,\phi,\xi,\theta)$
where $\sccv{i} M:T$ is a term and other components are as defined above.
Active configurations correspond to computations by the program (Player).
An active configuration can evolve into another active configuration through silent moves,
or into a passive one after the program performs an action.

The admissible transitions in, what we call the Program LTS, are shown in \Cref{fig:Program LTS}.

\begin{figure}
\begin{itemize}
    \item (P$\tau$) $(M_1,i,T,\Sigma_1,h_1,\gamma,\phi,\xi,\theta)\xra{\tau}(M_2,i,T,\Sigma_2,h_2,\gamma,\phi,\xi,\theta)$
    
    where $(M_1,h_1)\ra(M_2,h_2)$.

    \item (PA) $(V, i, T, \Sigma, h, \gamma, \phi, \xi, \theta)\xra{(\ol{A},\Sigma', \chi)}(\Sigma, h, \gamma\uplus\gamma'\uplus\gamma'', \phi\uplus\nu(A)\uplus \nu(\chi), \xi, \dom{\chi})$
    
    where $(A,\gamma')\in\AVal{\sccv{i} V:T}$, 
    $(\Sigma',h')=\rst{(\Sigma,h)}{\theta\cup\mu(A)}$,
    $(\chi,\gamma'')\in\AVal{\spccv{}h':\heap}$

    \item (PQF) $(K[fV], i, T, \Sigma, h, \gamma, \phi, \xi, \theta)\xra{(\ol{f}(A),\Sigma', \chi)}(\Sigma, h, \gamma\uplus\gamma'\uplus\gamma'', \phi\uplus\nu(A)\uplus\nu(\chi), \xi:(\Sigma\vd K:(\cdot;\cdot;T_2;i_3)\ra(\cdot;\cdot;T;i)), \dom{\chi})$
    
    where $f\in \FNames_{T_1\ra T_2, i_1}$,
    $(A,\gamma')\in\AVal{\sccv{i_2} V:T_1}$,
    $(\Sigma',h')=\rst{(\Sigma,h)}{\theta\cup\mu(A)}$,
    $(\chi,\gamma'')\in\AVal{\spccv{}h':\heap}$

    \item (PQB) $(K[\#b^{\delta}], i, T, \Sigma, h, \gamma, \phi, \xi, \theta)\xra{(\ol{\run b^\rho},\Sigma', \chi)}(\Sigma, h, \gamma\uplus\gamma'\uplus\gamma'', \phi\uplus\nu(\chi), (\Sigma\vd K:(\cdot;\cdot;T_1;i_1)\ra(\cdot;\cdot;T;i)):\xi, \dom{\chi})$

    where $b\in\BNames_{\Box (\Gamma\vd T_1), 1}$,
    $(\rho, \gamma')\in\AVal{\sccv{i_1}\delta:\Gamma}$,
    $(\Sigma',h')=\rst{(\Sigma,h)}{\theta\cup\mu(\rho)}$,
    $(\chi,\gamma'')\in\AVal{\spccv{}h':\heap}$

    \item (OA) $(\Sigma, h, \gamma, \phi, \xi :(\Sigma''\vd K:(\cdot;\cdot;T_1;i_1)\ra(\cdot;\cdot;T_2;i_2)),\theta)\xra{(A,\Sigma',\chi)}(K[A], i_2, T_2, \Sigma\cup\Sigma', \chi+h, \gamma, \phi\uplus\nu(A)\uplus\nu(\chi), \xi, \dom{\chi})$
    
    where $\spccv{i_1} A:T_1$,
    $\dom{\chi}\cap\dom{h} = \theta$, $\chi = \rst{\chi}{\theta\cup \mu(A)}$

    \item (OQF) $(\Sigma, h, \gamma, \phi, \xi, \theta)\xra{(f(A),\Sigma',\chi)}(VA, i_3, T_2, \Sigma\cup\Sigma', \chi+h,\gamma,\phi\uplus\nu(A)\uplus\nu(\chi),\xi,\dom{\chi})$
    
    where $f\in \FNames_{T_1\ra T_2, i_1}$,
    $\gamma(f)=V$,
    $\spccv{i_2} A:T_1$,
    $i_3=\max(i_1, i_2)$,
    $\dom{\chi}\cap\dom{h} = \theta$, $\chi = \rst{\chi}{\theta\cup \mu(A)}$

    \item (OQB) $(\Sigma, h, \gamma, \phi, \xi, \theta)\xra{(\run b^\rho,\Sigma',\chi)}((\#V)[\rho], i, T, \Sigma\cup\Sigma', \chi+h,\gamma,\phi\uplus\nu(\chi),\xi,\dom{\chi})$
    
    where $b\in \BNames_{\Box(\Gamma\vd T), }$,
    $\gamma(b)=V$,
    $\spccv{i} \rho:\Gamma$,
    $\dom{\chi}\cap\dom{h} = \theta$, $\chi = \rst{\chi}{\theta\cup \mu(A)}$
    
\end{itemize}

\begin{flushleft}
We use $\uplus$ to denote disjoint union.

\medskip

Given two heaps $h$ and $h'$,
we define $h+h'$ to be a heap with $\dom{h+h'}=\dom{h}\cup\dom{h'}$
such that $(h+h')(\ell)=h(\ell)$ if $\ell\in\dom{h}$,
and $(h+h')(\ell)=h'(\ell)$ otherwise.

\medskip

Given a set of locations $\theta$ and a heap $\sccv{}h:\heap$,
We define $\rst{(\Sigma,h)}{\theta}=(\Sigma',h')$ where
$h'$ is the minimal heap such that
(a) $h'\subseteq h$,
(b) $\theta\subseteq\dom{h'}$,
and (c) $h'(\ell)=\ell'$ implies $\ell'\in\dom{h'}$,
and $\Sigma'\subseteq\Sigma$ is the corresponding store context such that
$\spccv{} h':\heap$.
When store contexts are clear from the context, we sometimes simply write $\rst{h}{\theta}=h'$.
\end{flushleft}

\caption{Transition rules for Program LTS}
\label{fig:Program LTS}
\end{figure}

(P$\tau$) is a silent transition that embeds the operational semantics into the LTS.
(PA) concerns cases when the LTS has reached a value.
The value is then announced in abstracted form $A$ through the action $(\ol{A}, \Sigma, \chi)$,
where $\chi$ is the abstraction of the current heap $h$,
restricted to locations that have been revealed to the context.
(PQF) corresponds to reaching a $\beta$-redex
where the function is an FName $f$.
In other words, the program wants to call an unknown function provided earlier by the context.
This is represented by the action
$(\ol{f}(A), \Sigma, \chi)$, where $A$ abstracts the argument.
Meanwhile, the current evaluation context $K$ is stored on the stack $\xi$ for later use.
(PQB) is analogous to the previous case,
and handles running unknown pieces of code.
Note that a local substitution $\delta$ is always provided for all its free variables,
and this substitution is also converted into an abstract representation $\rho$.
Such a question is represented by the action $(\ol{\run b^\rho}, \Sigma, \chi)$.

The above (PA), (PQF), and (PQB) actions lead from active to passive configurations. 
Below we discuss transitions from passive to active configurations, 
which correspond to actions by the context.
In (OA) the context provides a value to a pending player question,
which is represented by the action $(A,\Sigma,\chi)$.
The stack $\xi$ is popped
and $A$ is plugged into the top evaluation context.
The $\dom{\chi}\cap\dom{h} = \theta$ condition
ensures that the context cannot modify the program's private locations,
while $\chi = \rst{\chi}{\sigma\cup \mu(A)}$ means that
the context should only share relevant locations.
(OQF) describes a situation where the context calls 
a function that was communicated to it previously as $f$, and applies it to an 
arbitrary abstract argument $A$, as illustrated by the action $(f(A),\Sigma,\chi)$.
The transition fetches the corresponding actual value $\gamma(f)$,
which is then scheduled for execution in the successor active configuration.
Finally, (OQB) covers the case when the context wants to run code $b$ under some substitution $\rho$.
This rule is analogous to the previous one and the action is $(\run b^\rho,\Sigma,\chi)$.

\subsection{Trace semantics}

To start defining the trace semantics,
we first need to define the initial configuration induced by an arbitrary term $M$.
To handle free global variables in $M$,
we define abstraction over global substitutions as
$\AVal{\Sigma;\cdot\vd\sigma:\Psi}\df
\{(\eta,\bigcup_u\gamma_u)\mid \dom{\eta}=\dom{\sigma},\allowbreak
\eta(u)=\#(\eta'(u)),
(\eta'(u),\gamma_u)\in\AVal{\sccv{1}\bx \sigma(u):\Box(\Gamma\vd T)} \text{ for all } u:(\Gamma\vd T)\in\Psi,
\cap_{u\in\dom{\eta}}\nu(\eta(u))=\emptyset\}$.
We use $\eta$ to refer to an \textit{abstract global substitution}
that maps global variables to $\#b$ of the corresponding type.

Reminiscent of \Cref{thm: CIU},
the initial configurations are parameterized by the term $M$,  
an abstract global substitution $\eta$,
an abstract local substitution $\rho$,
and an abstract heap $\chi$.
\begin{definition}[Initial configurations]\label{def:initial_configuration1}
    Given a term $\spgv{i} M : T$,
    abstract global substitution $\Sigma';\cdot\vd \eta:\Psi$,
    abstract local substitution $\spccv{i}\rho:\Gamma$,
    and abstract heap $\spccv{}\chi:\heap$
    such that $\Sigma'\supseteq\Sigma$,
    the program configuration is defined as
    $\cconf{M}{\eta,\rho,\chi}=(M[\eta][\rho], i, T, \Sigma', \chi, \emptyset, \nu(\eta)\uplus\nu(\rho)\uplus\nu(\chi), \cdot, \dom{\chi})$.
\end{definition}
The semantics will only contain traces that empty the stack, as these represent the convergent interactions.
Below we define the set of such traces $\Tr{\CC}$ that are generated from $\CC$.
\begin{definition}[Derived traces]
    Given two configurations $\CC$, $\CC'$ and action $\act$, we write $\CC\xRa{\act}\CC'$ if $\CC\xra{\tau}^*\CC''\xra{\act}\CC'$.
    Given a trace $\tr\ =\ \act_1\dots\act_n$, we write $\CC\xRa{\tr}\CC'$ if $\CC\xRa{\act_1}\CC''\xRa{\act_2}\cdots\xRa{\act_n}\CC'$.
    We write $\Tr{\CC}=\{\tr\mid\CC\xRa{\tr}(\Sigma, h, \gamma, \phi, \cdot, \theta)\}$.
    Equivalently, $\tr\in\Tr{\CC}$ iff $\CC\xRa{\tr}\CC'$ for some $\CC'$ and the number of (OA) actions minus the number of (PQ) actions in $\tr$ equals the size of $\xi$ in $\CC$.
\end{definition}
Finally, we define the intended trace semantics of Contextual MetaML terms.
\begin{definition}[Trace semantics]\label{dfn:semantics}
    Given a term $\spgv{i} M:T$, we define its trace semantics as
    \[
    \Tr{M}=\{(\eta,\rho,\chi,\tr)\mid
    \eta,\rho,\chi\text{ are as specified in \Cref{def:initial_configuration1},}
    \text{and }\tr\in \Tr{\cconf{M}{\eta, \rho, \chi}}\}
    \]
\end{definition}
In the next section, this semantics will be shown fully abstract with respect 
to $\lsmc$, and hence to $\lsm$ by \Cref{thm: CIU}.

\begin{example}\label{exm: LTS traces}
    In this example, we consider some terms and the traces that they can generate under our LTS.
    We omit the typing context $\Sigma$ and heap $\chi$ in the traces where possible for brevity.
\begin{itemize}
    \item Term: $\cdot;\cdot;\cdot\vd_1\bx x:\Box(x:\Int\vd\Int)$

    Trace: $\ol{b_1}\quad \run b_1^{1/x} \quad \ol{1}$
    \item Term: $\cdot; \cdot; y:\Box(x:\Int\vd\Int)\ra\Int \vd_1 y\ \bx x: \Int$

    Substitution: $y\mapsto f_1$

    Trace: $\ol{f_1(b_1)}\quad \run b_1^{1/x} \quad \ol{1} \quad 2\quad \ol{2}$
    \item Term: $\cdot;u:(x:\Int\vd\Int);y:\Int\vd_0 u^{y/x} + 3:\Int$
    
    Substitution: $u\mapsto \#b_1$, $y\mapsto 1$

    Trace: $\ol{\run b_1^{1/x}} \quad 5 \quad \ol{8}$

    \item Term: $\ell:\rf \Int;u:(\ \vd \Unit);\cdot\vd_0 u;\ell\mapsto !\ell+2;u:\Unit$
    
    Substitution: $u\mapsto \#b_1$

    Initial heap: $\ell\mapsto 0$

    Trace: $(\ol{\run b_1}, \bc{\ell\mapsto 0}) \quad ((), \bc{\ell\mapsto 1}) \quad (\ol{\run b_1}, \bc{\ell\mapsto 3}) \quad ((), \bc{\ell\mapsto 4}) \quad (\ol{()}, \bc{\ell\mapsto 4})$
\end{itemize}
\end{example}

\section{Full Abstraction}\label{sec: full abstraction}

In order to prove full abstraction we need to develop terminology and results about the interaction between 
active and passive configurations. In \Cref{def:initial_configuration1}, we showed how to assign
\emph{active} configurations to terms. Below we define \emph{passive} configurations corresponding 
to the evaluation context $K$, substitutions $\sigma$, $\delta$, and heap $h$
as given in \Cref{def: closed instances of use}.
\begin{definition}
    Given an evaluation context $\Sigma\vd K:(\cdot;\cdot;T_1;i_1)\ra (\cdot;\cdot;T_2;i_2)$,
    global substitution $\Sigma;\cdot\vd\sigma:\Psi$,
    local substitution $\sccv{i}\delta:\Gamma$,
    and heap $\sccv{} h:\heap$,
    let $(\eta,\gamma)\in\AVal{\Sigma;\cdot\vd\sigma:\Psi}$,
    $(\rho,\gamma')\in\AVal{\sccv{i}\delta:\Gamma}$,
    $(\Sigma',h')=\rst{(\Sigma,h)}{\mu(\eta)\cup\mu(\rho)}$,
    and $(\chi,\gamma'')\in\AVal{\spccv{}h':\heap}$,
    the context configuration is defined as
    $\cconf{K,\sigma,\delta,h}{\eta,\gamma,\rho,\gamma',\chi,\gamma''}=
    (\Sigma,h,\gamma\uplus\gamma'\uplus\gamma'',\nu(\eta)\uplus\nu(\rho)\uplus\nu(\chi), \cdot:(\Sigma\vd K:(\cdot;\cdot;T_1;i_1)\ra (\cdot;\cdot;T_2;i_2)),\dom{\chi})$.
\end{definition}

Given a sequence $\tr$ of actions, we write $\ol{\tr}$ for the sequence obtained by converting 
each action to one of the opposite polarity (i.e. context to program, program to context).
We write $\CC_1|\CC_2\da^\tr$ if either (a) $\tr\in\Tr{\CC_1}$ and $\ol{\tr}\; \act\in\mathbf{Tr}(\CC_2)$, 
or (b) $\tr\; \act\in\Tr{\CC_1}$ and $\ol{\tr}\in\mathbf{Tr}(\CC_2)$ for some (PA) action $\act$.

To show soundness, we need to show that a program configuration and a context configuration share a trace
iff the term obtained by combining them terminates.
\begin{theorem}[Correctness]\label{thm:correctness}
    For any term $\spgv{i}M:T$,
    evaluation context $\Sigma'\vd K:(\cdot;\cdot;i;T)\ra (\cdot;\cdot;i';T')$,
    global substitution $\Sigma';\cdot\vd\sigma:\Psi$,
    local substitution $\spccv{i}\delta:\Gamma$,
    and heap $\spccv{} h:\heap$
    such that $\Sigma'\supseteq\Sigma$,
    let $(\eta,\gamma)\in\AVal{\Sigma';\cdot\vd\sigma:\Psi}$,
    $(\rho,\gamma')\in\AVal{\spccv{i}\delta:\Gamma}$,
    $(\Sigma'',h')=\rst{(\Sigma',h)}{\mu(\eta)\cup\mu(\rho)}$,
    and $(\chi,\gamma'')\in\AVal{\Sigma'';\cdot;\cdot\vd h':\heap}$,
    then $(K[M[\sigma][\delta]],h)\Da$ iff
    there exists trace $\tr$ such that
    $\cconf{K,\sigma,\delta,h}{\eta,\gamma,\rho,\gamma',\chi,\gamma''}\mid\cconf{M}{\eta,\rho,\chi}\da^{\tr}$.
\end{theorem}

\begin{proof_sketch}
    We generalize the statement and prove that for any pair of ``compatible'' configurations,
    they share a trace iff their ``combination'' terminates.
    In particular, we show that such pairs cannot generate an infinite $\tau$-free trace
    by giving an upper bound on the length of such traces.
\end{proof_sketch}

\begin{example}
Here we revisit \Cref{exm: LTS traces} and show the context configurations that correspond to those terms and traces.
\begin{itemize}
    \item Term: $\cdot;\cdot;\cdot\vd_1\bx x:\Box(x:\Int\vd\Int)$

    Evaluation context: $\letbox{u}{\bullet}{(u^{1/x} + 1)}$

    Trace: $b_1\quad \ol{\run b_1^{1/x}} \quad 1 \quad \ol{2}$
    \item Term: $\cdot; \cdot; y:\Box(x:\Int\vd\Int)\ra\Int \vd_1 y\ \bx x: \Int$

    Evaluation context: $\bullet + 2$

    Substitution: $y\mapsto \la z.\letbox{u}{z}{(u^{1/x} + 1)}$

    Trace: $f_1(b_1)\quad \ol{\run b_1^{1/x}} \quad 1 \quad \ol{2} \quad 2\quad \ol{4}$

    \item Term: $\cdot;u:(x:\Int\vd\Int);y:\Int\vd_0 u^{y/x} + 3:\Int$
    
    Evaluation context: $\bullet + 2$

    Substitution: $u\mapsto x+4$, $y\mapsto 1$

    Trace: $\run b_1^{1/x} \quad \ol{5} \quad 8 \quad \ol{10}$

    \item Term: $\ell:\rf \Int;u:(\ \vd \Unit);\cdot\vd_0 u;\ell\mapsto !\ell+2;u:\Unit$
    
    Evaluation context: $\bullet$

    Substitution: $u\mapsto \ell:=!\ell+1$

    Initial heap: $\ell\mapsto 0$

    Trace: $(\run b_1, \bc{\ell\mapsto 0}) \quad (\ol{()}, \bc{\ell\mapsto 1}) \quad (\run b_1, \bc{\ell\mapsto 3}) \quad (\ol{()}, \bc{\ell\mapsto 4}) \quad ((), \bc{\ell\mapsto 4}) \quad \quad (\ol{()}, \bc{\ell\mapsto 4})$
\end{itemize}
\end{example}

\begin{theorem}[Soundness]\label{thm: soundness}
    Given two terms $\spgv{i} M_1,M_2:T$,
    if $\Tr{M_1}\subseteq\Tr{M_2}$,
    then $M_1\lsm M_2$.
\end{theorem}

To show completeness, we prove a definability property stating that,
for any well-behaved sequence of actions,
we can find a configuration that generates this trace only.

We say that a name is \textit{introduced} in an action
if it is used as an answer, as the argument in a (QF) action,
as the codomain of the substitution in a (QB) action,
or appears in $\chi$.
We say that a (PA) (resp. (OA)) action \textit{is an answer to} an (OQ) (resp. (PQ)) action
if it is the first (PA) (resp. (OA)) action
such that the number of (PA) (resp. (OA)) actions between them
equals the number of (OQ) (resp. (PQ)) actions between them.

Given $N_O,N_P\subseteq\Names$, an \textit{$N_O,N_P$-trace} is a sequence of actions such that
(i) Player and Opponent actions alternate;
(ii) no name in $N_P\cup N_O$ is introduced;
(iii) no name is introduced twice;
(iv) if a name is used as a question,
then it must either occur in the set $N_\_$ of opposite polarity
or have been introduced in a previous action of opposite polarity;
(v) an answer to a question must have the same type as the question is expecting; and
(vi) if $\chi$ and $\chi'$ belong to two neighbouring actions $\act$ and $\act'$,
then $\dom{\chi'}$ should be the minimal set of locations satisfying the following three conditions:
(a) $\dom{\chi}\subseteq\dom{\chi'}$,
(b) locations occurring in $\act'$ (if any) occur in $\dom{\chi'}$,
and (c) $\chi'(\ell)=\ell'$ implies $\ell'\in\dom{\chi'}$.

We say an $\emptyset,N_P$-trace is \textit{complete} if
(i) it starts with an Opponent action,
(ii) it ends with a (PA) action,
(iii) the number of (OA) actions equals the number of (PQ) actions plus one,
(iv) in any prefix of it the number of (OA) actions is at most the number of (PQ) actions plus one, and
(v) in any proper prefix of it the number of (PA) actions is at most the number of (OQ) actions.
Then the traces in $\Tr{\CC}$ are complete iff $\CC$ is a passive configuration
whose stack $\xi$ contains exactly one evaluation context.
In particular, if $\tr\in\Tr{\cconf{K,\sigma,\delta,h}{\eta,\gamma,\rho,\gamma',\chi,\gamma''}}$,
then $\tr$ must be a complete $\emptyset,\dom{\gamma}\uplus\dom{\gamma'}\uplus\dom{\gamma''}$-trace.

\begin{lemma}[Definability]\label{lem: definability}
    Given a complete $\emptyset,N_P$-trace $\tr$,
    there exists $\CC=(\Sigma,h,\gamma,N_P,\cdot : (\Sigma\vd K:(\cdot;\cdot;i_1;T_1)\ra (\cdot;\cdot;i_2;T_2)),\theta)$
    such that $\Tr{\CC}=\{\tr\}$ up to renamings that preserve $N_P$.
\end{lemma}

\begin{proof_sketch}
    Suppose $\tr\ =\ \act_1\dots\act_n$,
    and let $\tr_i\ =\ \act_i\dots\act_n$.
    We specify the configurations $\CC_i$
    such that $\Tr{\CC_i}=\{\tr_i\}$ up to renaming
    by backward induction on $i$.
    We create references for all the names,
    and set the references appropriately at each step
    so that only the designated actions can be taken.
    For P actions, $M_i$ should set the locations according to $\chi_i$
    and generate the corresponding player action.
    For O actions, the popped evaluation context should assert that the action is received as expected,
    assert that the shared locations are updated as expected,
    set the stored values to be the same as $h_{i+1}$,
    and return $M_{i+1}$.
\end{proof_sketch}

\begin{theorem}[Completeness]\label{thm: completeness}
    Given two terms $\spgv{i} M_1,M_2:T$,
    if $M_1\lsm M_2$,
    then $\Tr{M_1}\subseteq\Tr{M_2}$.
\end{theorem}

\begin{theorem}[Full Abstraction]\label{thm:fullAbstraction}
    Given two terms $\spgv{i} M_1,M_2:T$,
    $M_1\lsm M_2$
    iff $\Tr{M_1}\subseteq\Tr{M_2}$.
\end{theorem}

If we consider the fragment of Contextual MetaML that does not contain global variables,
box terms, or box types,
the resulting language is essentially simply-typed ML with higher-order references,
which we shall call HOS \cite{dreyerImpactHigherorderState2012,jaberCompleteTraceModels2021} here.
For an HOS term,
its trace semantics does not involve $\BNames$ or (QB) rules,
and coincides with its trace semantics in HOS as given by Laird \cite{lairdFullyAbstractTrace2007}.
Therefore, we have the following result.

\begin{theorem}[Conservative extension]\label{thm:conservative_extension}
    Contextual MetaML is a conservative extension of HOS with respect to contextual equivalence.
    In other words, two HOS terms are equivalent in HOS contexts
    iff they are equivalent in Contextual MetaML contexts.
\end{theorem}

\section{Examples}\label{sec: examples}
\begin{example}[Switching letbox order]\label{exm: switch letbox}
Let typing context $\Gamma_1$ be
\[x:\Box(\cdot\vd\Int),y:\Box(\cdot\vd\Int)\]
and terms $M_1, M_2$ be
\[\Gamma_1\vd_1 \letbox{u}{x}{\letbox{v}{y}{u;v}}:\Int\]
\[\Gamma_1\vd_1 \letbox{v}{y}{\letbox{u}{x}{u;v}}:\Int\]
respectively. Let $\rho=b_1/x,b_2/y$, $\eta=\cdot$, and $\chi=\cdot$, then
\[\begin{array}{rl}
\CC_{M_1}^{\eta,\rho,\chi} = & (\letbox{u}{b_1}{\letbox{v}{b_2}{u;v}},1,\Unit,\cdots)\\
\xra{\tau} & (\letbox{v}{b_2}{\#b_1;v},1,\Unit,\cdots)\\
\xra{\tau} & (\#b_1;\#b_2,1,\Unit,\cdots)\\
\end{array}\]
Similarly, $\CC_{M_2}^{\eta,\rho,\chi}\xra{\tau}^*(\#b_1;\#b_2,1,\Unit,\cdots)$.
So $\Tr{M_1}=\Tr{M_2}$, and therefore, $M_1\approx M_2$.

The above equivalence holds because of the call-by-name nature of the Letbox operation.
This is not to be confused with the fact that the argument of letbox is still supplied in a call-by-value way.
For example, let typing context $\Gamma_2$ be
\[x:\Unit\ra\Box(\cdot\vd\Int),y:\Unit\ra\Box(\cdot\vd\Int)\]
and terms $M_3, M_4$ be
\[\Gamma_2\vd_1 \letbox{u}{x()}{\letbox{v}{y()}{u;v}}:\Int\]
\[\Gamma_2\vd_1 \letbox{v}{y()}{\letbox{u}{x()}{u;v}}:\Int\]
respectively.
Let $\rho=f_1/x,f_2/y$, $\eta=\cdot$, and $\chi=\cdot$.
Then $\CC_{M_3}^{\eta,\rho,\chi}\xra{(\ol{f_1()},\cdot,\cdot)}\cdots$,
and $\CC_{M_4}^{\eta,\rho,\chi}\xra{(\ol{f_2()},\cdot,\cdot)}\cdots$.
So $M_3$ and $M_4$ are incomparable.
\end{example}

\begin{example}[Code duplication]\label{exm: code duplication}
Let terms $M_1, M_2, M_3, M_4$ be
\[
\begin{aligned}
x:\Unit\ra\Unit &\vd_1\letbox{u}{\bx x()}{u^{x/x};u^{x/x}}\\
x:\Unit\ra\Unit &\vd_1\letbox{u}{\bx z()}{u^{x/z};u^{x/z}}\\
x:\Unit\ra\Unit &\vd_1 x();x()\\
x:\Unit\ra\Unit &\vd_1 \letin{y}{x()}{y;y}
\end{aligned}
\]
respectively.
Then, intuitively, both $M_1$, $M_2$ and $M_3$ will call $x$ twice,
while $M_4$ will only call $x$ once.
Let $\rho=f_1/x$ and $\chi=\ell_1\mapsto 0$, then both $M_1$, $M_2$, and $M_3$ (but not $M_4$) have the following trace
(we omit $\Sigma$ in the action for brevity):
\[(\ol{f_1}(),\ell_1\mapsto 0)\; ((), \ell_1\mapsto 1)\; (\ol{f_1}(),\ell_1\mapsto 1)\; ((), \ell_1\mapsto 3)\; (\ol{()},\ell_1\mapsto 3)\]
On the other hand, $M_4$ (but not $M_1$, $M_2$ or $M_3$) has the trace
\[(\ol{f_1}(),\ell_1\mapsto 0)\; ((), \ell_1\mapsto 1)\; (\ol{()},\ell_1\mapsto 1)\]
Therefore, $M_1\approx M_2\approx M_3\not\approx M_4$.

When $x()$ is replaced with some long effectful program,
$\mathbf{letbox}$ allows type-safe code duplication.
This is one of the common usages of metaprogramming.
\end{example}

\begin{example}[Box binds local variables]\label{exm: box binding}
    Let $M_1,M_2,M_3,M_4$ be
    \[
    \begin{aligned}
    \vd_1&\la x.\la y.\bx x:\Int\ra\Int\ra\Box(x:\Int,y:\Int\vd\Int) \\
    \vd_1&\la y.\la x.\bx x:\Int\ra\Int\ra\Box(x:\Int,y:\Int\vd\Int) \\
    \vd_1&\la z_1.\la z_2.\bx x:\Int\ra\Int\ra\Box(x:\Int,y:\Int\vd\Int) \\
    \vd_1&\la x.\la y.\bx y:\Int\ra\Int\ra\Box(x:\Int,y:\Int\vd\Int) \\
    \end{aligned}
    \]
    respectively. Then $M_1\approx M_2\approx M_3$, and they are incomparable to $M_4$.
    This is because the free variables inside $\mathbf{box}$ are bound by $\mathbf{box}$ and revealed to the context via the type of the term.
    They are not bound by the preceding lambda abstractions,
    and will not be $\alpha$-converted when we rename those lambda abstractions.
    For example, the trace
    \[(\ol{f_1},\cdot,\cdot)\; (f_1(1),\cdot,\cdot)\; (\ol{f_2},\cdot,\cdot)\; (f_2(2),\cdot,\cdot)\;
    (\ol{b_1},\cdot,\cdot)\; (\run b_1^{3/x,4/y},\cdot,\cdot)\; (\ol{3},\cdot,\cdot)\]
    belongs to $\Tr{M_1}$, $\Tr{M_2}$, and $\Tr{M_3}$, while the trace
    \[(\ol{f_1},\cdot,\cdot)\; (f_1(1),\cdot,\cdot)\; (\ol{f_2},\cdot,\cdot)\; (f_2(2),\cdot,\cdot)\;
    (\ol{b_1},\cdot,\cdot)\; (\run b_1^{3/x,4/y},\cdot,\cdot)\; (\ol{4},\cdot,\cdot)\]
    belongs to $\Tr{M_4}$.
\end{example}

Below we use the imperative power function given by Calcagno et al. \cite{calcagnoClosedTypesSafe2003}
to demonstrate our model's ability to reason about programs with effectful types.
Here the transformed term is not equivalent to the original term
as computation related to $n$ has been brought forward.
However, the $\eta$-expansion of the transformed term is equivalent to the original term,
so the transformation is faithful in the sense that it never gives different results.
\begin{example}[Correctness of staging the imperative power function]\label{exm: correctness staging}
    Let $\power$ be the following term:
    \[
    \begin{aligned}
        &
        \begin{aligned}
            \vd_1\rec{f}{n}{\la xy.&\mathbf{if}\ n=0\ \mathbf{then}\ {y:=1}\\
            &\mathbf{else}\ (f\ (n-1)\ x\ y;y:=!y*x)}
        \end{aligned}\\
        &\colon\Int\ra \Int\ra\rf \Int\ra \Unit
    \end{aligned}
    \]

    Let $\powerstagedgen$ be the following term:
    \[
    \begin{aligned}
    &
    \begin{aligned}
        \vd_1\mathbf{rec\ }f(n).\la xy.&\mathbf{if}\ n=0\ \mathbf{then}\ \letbox{u}{y}{\bx (u^{y/y}:=1)}\\
        &
        \begin{aligned}
            \mathbf{else}\ &\letbox{u}{f\ (n-1)\ x\ y}\\
            &\letbox{v}{x}\\
            &\letbox{w}{y}\\
            &\bx (u^{x/x,y/y};w^{y/y}:=!w^{y/y}*v^{x/x})
        \end{aligned}
    \end{aligned}\\
    &
    \begin{aligned}
        \colon&\Int\ra \Box(x:\Int\vd\Int)\ra\Box(y:\rf\Int\vd\rf\Int)\\
        &\ra\Box(x:\Int,y:\rf\Int\vd\Unit)
    \end{aligned}
    \end{aligned}
    \]
    Let $\powerstaged$ be the following term:
    \[
    \begin{aligned}
    &\vd_1 \la n.\letbox{u}{\powerstagedgen\ n\ (\bx x)\ (\bx y)}{\la xy.u^{x/x,y/y}} \\
    &\colon\Int\ra\Int\ra\rf\Int\ra\Unit
    \end{aligned}
    \]
    Then $\powerstaged\lnapprox \power\approx \la nxy.\powerstaged\ n\ x\ y$.
\end{example}

\section{Related Work}\label{sec: related work}

\subparagraph*{Type systems for imperative MetaML}
Purely functional MetaML as proposed by Taha and Sheard \cite{tahaMetaMLMultistageProgramming2000}
was one of the first safe multi-stage programming languages with explicit annotations.
Since then, many type systems for imperative MetaML have been developed,
with a primary focus on designing a safe type system for higher-order references of code.
Notable examples include closed types \cite{calcagnoClosedTypesSafe2003},
refined environment classifiers \cite{kiselyovRefinedEnvironmentClassifiers2016a, isodaTypeSafeCodeGeneration2024},
Mint \cite{westbrookMintJavaMultistage2010},
and $\lambda^\oslash$ \cite{kameyamaShiftingStageStaging2011,kameyamaCombinatorsImpureHygienic2015}.
These type systems view scope extrusion as undesirable,
and prevent a piece of open code such as $\ag{x}$ (or, in our term, $\bx x$)
from being stored in a reference cell at all.
Another line of work such as MetaOCaml \cite{kiselyovMetaOCamlTheoryImplementation2023,kiselyovDesignImplementationBER2014}
and MacoCaml \cite{xieMacoCamlStagingComposable2023,chiangStagedCompilationModule2024}
gives up static type safety
and relies on a dynamic scope extrusion detection when dereferencing a cell at runtime.
In the former approach, open code cannot be stored.
In the latter approach, open code can be stored but cannot be used.
However, as we have shown in this paper,
there are situations where the ability to store and use open code is useful.
And scope extrusion does not necessarily lead to type unsafety,
as long as the type system reasons about free variables in code values properly.

\subparagraph*{Contextual type theory}
Early contextual modal type theories \cite{nanevskiContextualModalType2008,pientkaTypetheoreticFoundationProgramming2008,pientkaTypeTheoryDefining2019}
were originally developed for recursively analyzing open terms via higher-order pattern matching.
In those systems, there is a hard division between
the meta language that analyzes and object language that is analyzed.
Therefore, those type systems do not support running code,
which is essentially promoting the object language into the meta language.
More recently, M\oe bius \cite{jangMoebiusMetaprogrammingUsing2022}
is the first so-called homogeneous contextual type theory
where there is no distinction between the meta language and the object language,
and it supports running code.
However, its operational semantics is not type safe
as there is no guarantee of progress, i.e. a well-typed program might get stuck.
Layered modal type theory \cite{huLayeredModalType2024}
achieves type safety by restricting the depth of nested code,
and is most similar to our work.
There has been an extension of the layered approach to dependent types \cite{huDependentTypeTheory2025},
but this is orthogonal to our concern.
All of the above systems are purely functional;
in contrast, our contribution is the design of such a (layered modal type) system for higher-order references of arbitrary open code.
We drop the ability to pattern-match on code because this would make full abstraction trivial.
Any syntactically different code (e.g. $\bx (1+1)$ and $\bx 2$) could be distinguished by pattern matching,
and there would be no non-trivial equivalences between pieces of code.

\subparagraph*{Full abstraction for MetaML}
Taha \cite{tahaMultistageProgrammingIts1999} presented a few equational rules for purely functional MetaML,
which are further extended by Inoue and Taha \cite{inoueReasoningMultistagePrograms2016}.
Both equational systems are sound for contextual equivalence but not complete.
Berger and Tratt \cite{bergerProgramLogicsHomogeneous2015} gave a Hoare-style program logic
for typed purely functional MetaML that is fully abstract.
Inoue and Taha \cite{inoueReasoningMultistagePrograms2016}\footnote{In that paper, they believed
``imperative hygienic MSP (multi-stage programming) is not yet ready for an investigation like this ...
The foundations for imperative hygienic MSP have not matured to the level of the functional theory that we build upon here.''
Our type system solves their concern.} gave a fully abstract characterization
of contextual equivalence for untyped functional MetaML
using the notion of applicative bisimulation \cite{abramskyFullAbstractionLazy1993,gordonBisimilarityTheoryFunctional1999}.
We are the first to give a fully abstract model for imperative MetaML.

\subparagraph*{Operational game semantics}
Game semantics has been a powerful tool to obtain fully abstract models for a variety of programming languages.
Many such models are categorical constructions \cite{abramskyFullAbstractionPCF2000,hylandFullAbstractionPCF2000,rotStepsTraces2021,levyTransitionSystemsGames2014}
in the sense that programs are interpreted denotationally as morphisms (typically strategies) and types as objects (typically games) of a suitable category.
In contrast, the more recent operational game semantics
\cite{lairdFullyAbstractTrace2007,jaberCompleteTraceModels2021,lassenTypedNormalForm2007,borthelleAbstractCertifiedAccount2025}
provide a more intuitive operationally-driven framework.
Our trace model can be thought of as an extension of Laird's
model for higher-order references~\cite{lairdFullyAbstractTrace2007}.

\section{Future Work}\label{sec: future work}
Contextual MetaML only has two layers.
Although this already covers most practical use cases of generative metaprogramming,
there are scenarios where more than two layers are needed.
Hu and Pientka \cite{huLayeredApproachIntensional2025} presented a type-safe layered modal type theory
generalized to $n$ layers.
This generalization requires more work
than just letting the $i$ in the typing judgment range over natural numbers.
For example, the distinction between local and global variables should be generalized so that 
all variables $x^k$ carry an annotation $0\le k<n$
denoting the layer it is from.
And the contextual types should be of the form $\Box_k(\Gamma\vd_j T)$
that reads ``a term from layer $k$ that describes a code from layer $j$'' for $0\le j<k<n$.
We believe that if we generalize our system to $n$ layers in this way,
we can retain the type safety proof by following Hu and Pientka \cite{huLayeredApproachIntensional2025}.
However, adapting the CIU and full abstraction proofs will require some more work.
In CIU, for example,
we need to define the eligible closing substitutions for a free variable $x^k$
for any $k$.

\subparagraph*{\bf CC-BY statement} For the purpose of Open Access, the author has applied a CC-BY public copyright licence to any Author Accepted Manuscript (AAM) version arising from this submission.

\bibliography{refs-lipics}

\newpage
\appendix
\section{Proofs}
\subsection{Proofs for \Cref{sec: language}}
\begin{lemma}[Lifting]\label{lem: lifting}
    If $\Sigma;\Psi;\Gamma\vd_0 M:T$, then $\Sigma;\Psi;\Gamma\vd_1 M:T$.
\end{lemma}
\begin{proof}
    Any typing derivation at layer $i=0$ is also a valid typing derivation at layer $i=1$.
\end{proof}

\begin{lemma}[Context typing]\label{lem: context typing}
    If $\sccv{} C:(\Psi';\Gamma';i';T')\ra(\Psi;\Gamma;i;T)$,
    and $\Sigma;\Psi';\Gamma'\vd_{i'} M:T'$,
    then $\Sigma;\Psi;\Gamma\vd_i C[M]:T$.
\end{lemma}

\begin{proof}
    By induction on $C$.
\end{proof}

\begin{lemma}[Local substitution]\label{lem: local substitution}
    The following statements hold:
    \begin{itemize}
        \item If $\Sigma;\Psi;\Gamma'\vd_i \delta':\Gamma''$
        and $\Sigma;\Psi;\Gamma\vd_i \delta:\Gamma'$,
        then $\Sigma;\Psi;\Gamma\vd_i \delta'[\delta]:\Gamma''$.
        \item If $\Sigma;\Psi;\Gamma'\vd_i M:T$
        and $\Sigma;\Psi;\Gamma\vd_i \delta:\Gamma'$,
        then $\Sigma;\Psi;\Gamma\vd_i M[\delta]:T$.
    \end{itemize}
    
\end{lemma}
\begin{proof}
    We prove the two statements by mutual structural induction on $M$ and $\delta'$.
    Most cases are straightforward. We show several interesting cases here.
    \begin{itemize}
        \item For non-empty $\delta'$,
        we have $\Sigma;\Psi;\Gamma'\vd_i \delta'(x):\Gamma''(x)$ for all $x\in\dom{\delta'}$.
        As $\delta'(x)$ is structurally smaller than $\delta'$,
        we can use the induction hypothesis on $\delta'(x)$,
        and have $\Sigma;\Psi;\Gamma\vd_i \delta'(x)[\delta]:\Gamma''(x)$.
        Recalling that $\delta'[\delta](x)=\delta'(x)[\delta]$,
        we have $\Sigma;\Psi;\Gamma\vd_i \delta'[\delta](x):\Gamma''(x)$.
        Therefore, $\Sigma;\Psi;\Gamma\vd_i \delta'[\delta]:\Gamma''$.
        \item For $M=u^{\delta''}$,
        we have $\Psi(u)=(\Gamma''\vd T)$
        and $\Sigma;\Psi;\Gamma'\vd_i \delta'':\Gamma''$ by (Tm-GVar).
        As $\delta''$ is structurally smaller than $M$,
        we can use the induction hypothesis on $\delta''$,
        and have $\Sigma;\Psi;\Gamma\vd_i \delta''[\delta]:\Gamma''$.
        By (Tm-GVar) again, we have
        $\Sigma;\Psi;\Gamma\vd_i u^{\delta''[\delta]}:T$.
        This is exactly $\Sigma;\Psi;\Gamma\vd_i M[\delta]:T$.
        \item For $M=\bx M'$, we have $i=1$, $T=\Box(\Gamma''\vd T')$,
        and $\Sigma;\Psi;\Gamma''\vd_0 M':T'$ by (Tm-Box).
        By (Tm-Box) again, we have
        $\Sigma;\Psi;\Gamma\vd_1 \bx M':\Box(\Gamma''\vd T')$.
        This is exactly $\Sigma;\Psi;\Gamma\vd_i M[\delta]:T$.
    \end{itemize}
\end{proof}

\begin{lemma}[Global substitution]\label{lem: global substitution}
    The following statements hold:
    \begin{itemize}
        \item If $\Sigma;\Psi';\Gamma\vd_i \delta:\Gamma'$ and
        $\Sigma;\Psi\vd \sigma:\Psi'$,
        then $\Sigma;\Psi;\Gamma\vd_i \delta[\sigma]:\Gamma'$.
        \item If $\Sigma;\Psi';\Gamma\vd_i M:T$ and
        $\Sigma;\Psi\vd \sigma:\Psi'$,
        then $\Sigma;\Psi;\Gamma\vd_i M[\sigma]:T$.
    \end{itemize}
\end{lemma}
\begin{proof}
    We prove the two statements by mutual structural induction on $M$ and $\sigma$.
    Most cases are either straightforward or similar to above. We only show one interesting case.
    \begin{itemize}
        \item For $M=u^{\delta'}$,
        we have $\Psi'(u)=(\Gamma'\vd T)$
        and $\Sigma;\Psi';\Gamma\vd_i \delta':\Gamma'$ by (Tm-GVar).
        As $\delta'$ is structurally smaller than $M$,
        we can use the induction hypothesis on $\delta'$,
        and have $\Sigma;\Psi;\Gamma\vd_i \delta'[\sigma]:\Gamma'$.

        Since $\Sigma;\Psi\vd \sigma:\Psi'$,
        by (GSubst-Cons) we have
        $\Sigma;\Psi;\Gamma'\vd_0\sigma(u):T$.
        By \Cref{lem: lifting,lem: local substitution},
        we have $\Sigma;\Psi;\Gamma\vd_i \sigma(u)[\delta'[\sigma]]:T$.
        This is exactly $\Sigma;\Psi;\Gamma\vd_i M[\sigma]:T$.
    \end{itemize}
\end{proof}

We say a term $M$ is a \textit{redex} if it occurs on the left-hand side of a reduction rule apart from (Ktx).

\begin{theorem}[Preservation, \Cref{thm: preservation}]
    If $\spgv{i} M:T$, $\spgv{} h:\heap$, and $(M,h)\ra(M',h')$,
    then there exists $\Sigma'\supseteq\Sigma$ such that
    $\Sigma';\Psi;\Gamma\vd_i M':T$, and $\Sigma';\Psi;\Gamma\vd h':\heap$.
\end{theorem}
\begin{proof}
    Since $(M,h)\ra(M',h')$, we have $M=K[M_1]$
    for some evaluation context $\sccv{} K:(\Psi_1;\Gamma_1;i_1;T_1)\ra(\Psi;\Gamma;i;T)$
    and redex $\Sigma;\Psi_1;\Gamma_1\vd_{i_1}M_1:T_1$,
    and $(M_1,h)\ra(M_1',h')$ by a reduction rule apart from (Ktx).
    It suffices to show that there exists $\Sigma'\supseteq\Sigma$ such that
    $\Sigma';\Psi_1;\Gamma_1\vd_i M_1':T$
    and $\Sigma';\Psi_1;\Gamma_1\vd h':\heap$.
    We prove by case analysis on the reduction rule used.
    Most cases are straightforward,
    so we only show several interesting cases here.
    \begin{itemize}
        \item ($\beta$). Then $M_1=(\la x.M_2)V$, $M_1'=M_2[V/x]$, and $h'=h$.
        By (Tm-App) and (Tm-Abs),
        we have $\Sigma;\Psi_1;\Gamma_1,x:T_2\vd_{i_1}M_2:T_1$
        and $\Sigma;\Psi_1;\Gamma_1\vd_{i_1}V:T_2$.
        By \Cref{lem: local substitution},
        we have $\Sigma;\Psi_1;\Gamma_1\vd_{i_1}M_2[V/x]:T_1$.
        \item (Ref). Then $M_1=\rf V$, $M_1'=\ell$, $T_1=\rf T_2$, and $h'=h,\ell\mapsto V$
        for some fresh $\ell$.
        Let $\Sigma'=\Sigma,\ell:T_1$.
        By (Tm-Loc), we have $\Sigma';\Psi_1;\Gamma_1\vd_{i_1}M_1':T_1$.
        By (Heap), we have $\Sigma';\Psi_1;\Gamma_1\vd h':\heap$.
        \item (LetBox). Then $i_1=1$, $M_1=\letbox{u}{\bx M_2}{M_3}$,
        $M_1'=M_3[M_2/u]$, and $h'=h$.
        By (Tm-LetBox),
        we have $\Sigma;\Psi_1;\Gamma_1\vd_{1}\bx M_2:\Box(\Gamma_2\vd T_2)$
        and $\Sigma;\Psi_1,u:(\Gamma_2\vd T_2);\Gamma_1\vd_{1}M_3:T_1$.
        Using (Tm-Box) on the former,
        we have $\Sigma;\Psi_1;\Gamma_2\vd_{0}M_2:T_2$.

        By \Cref{lem: global substitution},
        we have $\Sigma;\Psi_1;\Gamma_1\vd_{1}M_3[M_2/u]:T_1$.
    \end{itemize}
\end{proof}

\begin{theorem}[Progress, \Cref{thm: progress}]
    If $\sccv{i} M:T$ and $\sccv{} h:\heap$,
    then either $M$ is a value, or there exist $M'$, $h'$ such that
    $(M,h)\ra(M',h')$.
\end{theorem}
\begin{proof}
    We prove by structural induction on $M$, and only list several interesting cases here.
    Notice that since the variable contexts are empty,
    $M$ cannot be $x$ or $u^\delta$.
    \begin{itemize}
        \item $M=\bx M_1$. Then $M$ is a value.
        \item $M=\letbox{u}{M_1}{M_2}$. By the induction hypothesis,
        either $M_1$ is a value, or there exist $M_1'$, $h'$ such that
        $(M_1,h)\ra(M_1',h')$.

        If $M_1$ is a value, by \Cref{lem: canonical forms},
        we have $M_1=\bx M_3$. Use the (LetBox) rule,
        we have $(M,h)\ra(M_2[M_3/u],h)$.

        If $(M_1,h)\ra(M_1',h')$, then we take $K=\letbox{u}{\bullet}{M_2}$,
        and have $(M,h)\ra(\letbox{u}{M_1'}{M_2},h')$ by the (Ktx) rule.
    \end{itemize}
\end{proof}

\subsection{Proofs for \Cref{sec: ciu}}

\begin{lemma}[CIU Basics, \Cref{lem: CIU basics} Part 1]
    Given $\spgv{i} M_1,M_2:t_1$, if $M_1\lsmc M_2$,
    then $C[M_1]\lsmc C[M_2]$ whenever the terms are typable and $C$ is one of the following forms:
    \begin{itemize}
        \renewcommand\labelitemi{--}
        \item $\bullet M$, \quad $\bullet \oplus M$, \quad $\iif{\bullet}{M}{M'}$, \quad
              $\rf \bullet$, \quad $!\bullet$, \quad $\bullet := M$, \quad
              $\letbox{u}{\bullet}{M}$,
        \item $M\bullet$, \quad $M\oplus\bullet$,
        \item $\iif{M}{\bullet}{M'}$, \quad $\iif{M}{M'}{\bullet}$,
    \end{itemize}
\end{lemma}

\begin{proof}
    We prove the first case in each category, and the rest are similar.
    \begin{itemize}
        \renewcommand\labelitemi{--}
        \renewcommand\labelitemii{*}
        \item Below we prove $M_1M\lsmc M_2M$.
              Suppose there exist evaluation context $K$, global substitution $\sigma$,
              local substitution $\delta$, and heap $h$
              that satisfy the conditions in \Cref{def: closed instances of use}
              and $(K[(M_1M)[\sigma][\delta]],h)\Da$.
              Let $K'=K[\bullet M[\sigma][\delta]]$,
              then we can rewrite the above supposition as $(K'[M_1[\sigma][\delta]],h)\Da$.
              Since $M_1\lsmc M_2$, we have $(K'[M_2[\sigma][\delta]],h)\Da$.
              This is exactly $(K[(M_2M)[\sigma][\delta]],h)\Da$.
        \item Below we prove $MM_1 \lsmc MM_2$.
              Again suppose there exist evaluation context $K$, global substitution $\sigma$,
              local substitution $\delta$, and heap $h$
              that satisfy the conditions in \Cref{def: closed instances of use}
              and $(K[(MM_1)[\sigma][\delta]],h)\Da$.
              We prove $(K[(MM_2)[\sigma][\delta]],h)\Da$
              by induction on the number of steps $k$ in the reduction $(K[(MM_1)[\sigma][\delta]],h)\Da$.
              \begin{itemize}
                \item $k=0$. Then $M[\sigma][\delta]$ must be a value $V$.
                Let $K'=K[V\bullet]$, and the rest is similar to the previous case.
                \item Suppose the result holds for $k-1$; we prove it for $k$.
                If $M[\sigma][\delta]$ is a value, the argument is the same as the previous case.
                Otherwise, $(M[\sigma][\delta],h)\ra(M',h')$,
                and $(K[M'M_1[\sigma][\delta]],h)\Da$ in $k-1$ steps.
                As $M[\sigma][\delta]$ is a closed term, so is $M'$.
                So we can rewrite the above statement as $(K[(M'M_1)[\sigma][\delta]],h)\Da$ in $k-1$ steps.
                By the induction hypothesis, we have $(K[(M'M_2)[\sigma][\delta]],h)\Da$.
                Therefore, $(K[(MM_2)[\sigma][\delta]],h)\Da$.
              \end{itemize}
        \item Below we prove $\iif{M}{M_1}{M'}\lsmc \iif{M}{M_2}{M'}$.
              Again suppose there exist evaluation context $K$, global substitution $\sigma$,
              local substitution $\delta$, and heap $h$
              that satisfy the conditions in \Cref{def: closed instances of use}
              and $(K[(\iif{M}{M_1}{M'})[\sigma][\delta]],h)\Da$.
              Based on \Cref{lem: canonical forms},
              we discuss what $M[\sigma][\delta]$ is.
              \begin{itemize}
                \item If $M[\sigma][\delta]$ is not a value, then the rest of the argument is similar to the previous case.
                \item If $M[\sigma][\delta]=\nb{0}$,
                then $(K[(\iif{M}{M_1}{M'})[\sigma][\delta]],h)\ra (K[M'[\sigma][\delta]],h)$.
                So $(K[M'[\sigma][\delta]],h)\Da$.
                Notice also that $(K[(\iif{M}{M_2}{M'})[\sigma][\delta]],h)\ra (K[M'[\sigma][\delta]],h)$.
                Therefore, $(K[(\iif{M}{M_2}{M'})[\sigma][\delta]],h)\Da$.
                \item If $M[\sigma][\delta]=\nb{n}$ for some $n\neq 0$,
                then $(K[(\iif{M}{M_1}{M'})[\sigma][\delta]],h)\ra (K[M_1[\sigma][\delta]],h)$.
                Since $M_1\lsmc M_2$, we have $(K[M_2[\sigma][\delta]],h)\Da$.
                Therefore, $(K[(\iif{M}{M_2}{M'})[\sigma][\delta]],h)\Da$.
              \end{itemize}
    \end{itemize}
\end{proof}

\begin{lemma}[CIU Abstraction, \Cref{lem: CIU basics} Part 2]\label{lem: CIU abstraction}
    Given $\spgv{i} M_1,M_2:T$, if $M_1\lsmc M_2$,
    then $\la x.M_1\lsmc \la x.M_2$.
\end{lemma}
\begin{proof}
    Suppose there exist evaluation context $K$, global substitution $\sigma$,
    local substitution $\delta$, and heap $h$
    that satisfy the conditions in \Cref{def: closed instances of use}
    and $(K[(\la x.M_1)[\sigma][\delta]],h)\Da$.
    We need to show that $(K[(\la x.M_2)[\sigma][\delta]],h)\Da$.
    For simplicity, we write $M_j^{\sigma,\delta}=M_j[\sigma][\delta,x/x]$ for $j=1,2$.
    We strengthen the claim and prove that $(M[\la x.M_1^{\sigma,\delta}/z],h)\Da$
    implies $(M[\la x.M_2^{\sigma,\delta}/z],h)\Da$ for any $M$.
    We prove by induction on the number of steps $k$ in the reduction $(M[\la x.M_1^{\sigma,\delta}/z],h)\Da$.
    \begin{itemize}
        \item $k=0$. Then $M[\la x.M_1^{\sigma,\delta}/z]$ is a value.
        Induction on $M$ shows that $M[\la x.M_2^{\sigma,\delta}/z]$ is also a value.
        \item Suppose result for $k-1$, we prove for $k$. Suppose
        $(M[\la x.M_1^{\sigma,\delta}/z],h)\ra(M',h')$,
        and $(M',h')\Da$ in $k-1$ steps.
        Consider whether the redex in the above first step reduction is created by the substitution.
        \begin{itemize}
            \item If not, then $(M[\la x.M_1^{\sigma,\delta}/z],h)\ra (M'[\la x.M_1^{\sigma,\delta}/z],h')$,
            and $M'[\la x.M_1^{\sigma,\delta}/z],h')\Da$ in $k-1$ steps.
            By the induction hypothesis, we have $(M'[\la x.M_2^{\sigma,\delta}/z],h')\Da$.
            Since $(M[\la x.M_2^{\sigma,\delta}/z],h)\ra (M'[\la x.M_2^{\sigma,\delta}/z],h')$,
            we have $(M[\la x.M_2^{\sigma,\delta}/z],h)\Da$.
            \item If yes, then the only possibility is that $M=K'[zV]$.
            Then $(K'[M_1^{\sigma,\delta}[V/x]][\la x.M_1^{\sigma,\delta}/z],h)\Da$ in $k-1$ steps.
            By the induction hypothesis, we have
            $(K'[M_1^{\sigma,\delta}[V/x]][\la x.M_2^{\sigma,\delta}/z],h)\Da$.
            We can rewrite the above statement as
            $(K'[M_1[\sigma][\delta,V/x]][\la x.M_2^{\sigma,\delta}/z],h)\Da$.
            Since $M_1\lsmc M_2$, we have $(K'[M_2[\sigma][\delta,V/x]][\la x.M_2^{\sigma,\delta}/z],h)\Da$.
            Therefore, $(M[\la x.M_2^{\sigma,\delta}/z],h)\Da$.
        \end{itemize}
    \end{itemize}
\end{proof}

\begin{lemma}[CIU Recursion, \Cref{lem: CIU basics} Part 3]
    Given $\spgv{i} M_1,M_2:T$, if $M_1\lsmc M_2$,
    then $\rec{y}{x}{M_1}\lsmc \rec{y}{x}{M_2}$.
\end{lemma}
\begin{proof}
    Suppose there exist evaluation context $K$, global substitution $\sigma$,
    local substitution $\delta$, and heap $h$
    that satisfy the conditions in \Cref{def: closed instances of use}
    and $(K[(\rec{y}{x}{M_1})[\sigma][\delta]],h)\Da$.
    We need to show that $(K[(\rec{y}{x}{M_2})[\sigma][\delta]],h)\Da$.
    For simplicity, we write $M_j^{\sigma,\delta}=M_j[\sigma][\delta,x/x,y/y]$ for $j=1,2$.
    We strengthen the claim and prove that $(M[\rec{y}{x}{M_1^{\sigma,\delta}}/z],h)\Da$
    implies $(M[\rec{y}{x}{M_2^{\sigma,\delta}}/z],h)\Da$ for any $M$.
    We prove by induction on the number of steps $k$ in the reduction $(M[\rec{y}{x}{M_1^{\sigma,\delta}}/z],h)\Da$.
    Most cases are similar to \Cref{lem: CIU abstraction}.
    We only show the case where the first step redex is created by the substitution.

    Suppose $M=K'[zV]$, then $(K'[M_1^{\sigma,\delta}[V/x,\rec{y}{x}{M_1^{\sigma,\delta}}/y]][\rec{y}{x}{M_1^{\sigma,\delta}}/z],h)\Da$ in $k-1$ steps.
    Notice that the above term can be rewritten as $(K'[M_1^{\sigma,\delta}[V/x,z/y]][\rec{y}{x}{M_1^{\sigma,\delta}}/z],h)$.
    Therefore, by the induction hypothesis, we have
    $(K'[M_1^{\sigma,\delta}[V/x,z/y]][\rec{y}{x}{M_2^{\sigma,\delta}}/z],h)\Da$.
    Rewrite the above statement as
    $(K'[M_1[\sigma][\delta,V/x,z/y]][\rec{y}{x}{M_2^{\sigma,\delta}}/z],h)\Da$.
    Since $M_1\lsmc M_2$, we have
    $(K'[M_2[\sigma][\delta,V/x,z/y]][\rec{y}{x}{M_2^{\sigma,\delta}}/z],h)\Da$.
    Therefore, $(M[\rec{y}{x}{M_2^{\sigma,\delta}}/z],h)\Da$.
    
\end{proof}

\begin{lemma}[CIU LetBox, \Cref{lem: CIU basics} Part 4]
    Given $\spgv{i} M_1,M_2:T$, if $M_1\lsmc M_2$,
    then $\letbox{u}{M}{M_1}\lsmc \letbox{u}{M}{M_2}$.
\end{lemma}
\begin{proof}
    Suppose there exist evaluation context $K$, global substitution $\sigma$,
    local substitution $\delta$, and heap $h$
    that satisfy the conditions in \Cref{def: closed instances of use}
    and $(K[(\letbox{u}{M}{{M_1}})[\sigma][\delta]],h)\Da$.
    We need to show that $(K[(\letbox{u}{M}{{M_2}})[\sigma][\delta]],h)\Da$.
    For simplicity, we write $M^{\sigma,\delta}=M[\sigma][\delta]$,
    and $M_j^{\sigma,\delta}=M_j[\sigma,u/u][\delta]$ for $j=1,2$.
    We prove by induction on the number of steps $k$ in the reduction
    $(K[\letbox{u}{M^{\sigma,\delta}}{{M_1^{\sigma,\delta}}}],h)\Da$.
    \begin{itemize}
        \item If $k=0$, then $K[\letbox{u}{M^{\sigma,\delta}}{{M_1^{\sigma,\delta}}}]$ is a value.
        Induction on $K$ shows that
        $K[\letbox{u}{M^{\sigma,\delta}}{{M_2^{\sigma,\delta}}}]$ is also a value.
        \item Suppose result for $k-1$, we prove for $k$. Consider whether $M^{\sigma,\delta}$ is a value.
        \begin{itemize}
            \item If yes, by \Cref{lem: canonical forms},
            we have $M^{\sigma,\delta}=\bx M'$.
            Then $(K[M_1^{\sigma,\delta}[M'/u]],h)\Da$ in $k-1$ steps.
            This can be rewritten as $(K[M_1[\sigma,M'/u][\delta]],h)\Da$.
            Since $M_1\lsmc M_2$, we have $(K[M_2[\sigma,M'/u][\delta]],h)\Da$.
            Therefore, $(K[\letbox{u}{M^{\sigma,\delta}}{{M_2^{\sigma,\delta}}}],h)\Da$.
            \item If not, then $(M^{\sigma,\delta},h)\ra(M',h')$,
            and $(K[\letbox{u}{M'}{{M_1^{\sigma,\delta}}}],h')\Da$ in $k-1$ steps.
            As $M^{\sigma,\delta}$ is a closed term, so is $M'$.
            So we can rewrite the above statement as
            $(K[\letbox{u}{(M')^{\sigma,\delta}}{{M_1^{\sigma,\delta}}}],h')\Da$ in $k-1$ steps.
            By the induction hypothesis, we have
            $(K[\letbox{u}{(M')^{\sigma,\delta}}{{M_2^{\sigma,\delta}}}],h')\Da$.
            Therefore, $(K[\letbox{u}{M^{\sigma,\delta}}{{M_2^{\sigma,\delta}}}],h)\Da$.
        \end{itemize}
    \end{itemize}
\end{proof}

\begin{lemma}\label{lem: beta preserve}
    Given closed box-free term $\sccv{i} M:T$,
    we have $(\la x.M)\ ()\approx M$.
\end{lemma}

\begin{proof}
    We need to show that $(C[(\la x.M)\ ()],\cdot)\Da$ iff $(C[M],\cdot)\Da$.
    We strengthen the claim and prove that
    $(N[(\la x.M)\ ()/z],h)\Da$ iff $(N[M/z],h)\Da$ for any closed term $N$ and heap $h$.
    \begin{itemize}
        \item We first prove the direction from left to right
        by induction on the number of steps $k$ in the reduction $(N[(\la x.M)\ ()/z],h)\Da$.
        Again most cases are similar to \Cref{lem: CIU abstraction},
        and we only show the case where the first step redex is created by the substitution.
        Suppose $N=K[z]$. Then $(K[(\la x.M)\ ()][(\la x.M)\ ()/z],h)\Da$ in $k$ steps.
        Therefore, $(K[M][(\la x.M)\ ()/z],h)\Da$ in $k-1$ steps.
        By the induction hypothesis, we have $(K[M][M/z],h)\Da$.
        This is exactly $(N[M/z],h)\Da$.
        \item We then prove the direction from right to left
        by induction on the number of steps $k$ in the reduction $(N[M/z],h)\Da$.
        Again most cases are similar to \Cref{lem: CIU abstraction},
        and we only show the case where the first step redex is created by the substitution.
        \begin{itemize}
            \item Suppose $N=K[z]$, and $M\ra M'$. Then $(K[M'][M/z],h)\Da$ in $k-1$ steps.
            By the induction hypothesis, we have $(K[M'][(\la x.M)\ ()/z],h)\Da$.
            Therefore, $(K[(\la x.M)\ ()][(\la x.M)\ ()/z],h)\Da$. Namely, $(N[(\la x.M)\ ()/z],h)\Da$.
            \item Suppose $N=K[zV]$, and $M=\la y.M'$. Then $(K[M'[V/y]][M/z],h)\Da$ in $k-1$ steps.
            By the induction hypothesis, we have $(K[M'[V/y][(\la x.M)\ ()/z],h)\Da$.
            Therefore, $(K[(\la x.M)\ ()V][(\la x.M)\ ()/z],h)\Da$. Namely, $(N[(\la x.M)\ ()/z],h)\Da$.
            \item Other cases are similar.
        \end{itemize}
    \end{itemize}
\end{proof}

\begin{lemma}[CIU Box, \Cref{lem: CIU box}]
    Given $\spgv{0} M_1,M_2:T$, if $M_1\lsmc M_2$,
    then $\bx M_1\lsmc \bx M_2$.
\end{lemma}
\begin{proof}
    Suppose there exist evaluation context $K$, global substitution $\sigma$,
    local substitution $\delta$, and heap $h$
    that satisfy the conditions in \Cref{def: closed instances of use}
    and $(K[(\bx M_1)[\sigma][\delta]],h)\Da$.
    We need to show that $(K[(\bx M_2)[\sigma][\delta]],h)\Da$.
    For simplicity, we write $M_j^\sigma=M_j[\sigma]$ for $j=1,2$.
    We strengthen the claim and prove that $(M[\bx M_1^\sigma/z],h)\Da$
    implies $(M[\bx M_2^\sigma/z],h)\Da$ for any $M$.
    We prove by induction on the number of steps $k$ in the reduction $(M[\bx M_1^\sigma/z],h)\Da$.
    Again most cases are similar to \Cref{lem: CIU abstraction},
    and we only show the case where the first step redex is created by the substitution.

    Suppose $M=K'[\letbox{u}{z}{M'}]$, then $(K'[M'[M_1^\sigma/u]][\bx M_1^\sigma/z],h)\Da$ in $k-1$ steps.
    By the induction hypothesis, we have
    $(K'[M'[M_1^\sigma/u]][\bx M_2^\sigma/z],h)\Da$.
    Let $M''=K'[M'][\bx M_2^\sigma/z]$,
    then we can rewrite the above statement as $(M''[M_1^\sigma/u],h)\Da$.
    At this point, our goal is to change the occurrence of $M_1$ in the above term to $M_2$
    using the fact that $M_1\lsmc M_2$.
    But we cannot do this directly because $M_1$ does not appear in an evaluation context.

    To solve this problem, we utilise \Cref{lem: CIU abstraction}
    and rewrite the term so that $\la y.M_1[\sigma][\delta_k]$
    appears in an evaluation context for all the necessary $\delta_k$s.
    The intuition behind this approach is that a function taking unit argument mimics a box,
    and applying such function to unit mimics unboxing.
    Suppose $\spccv{1} M'':T_1$.
    Recall that $\sccv{1} \bx M_1^\sigma:\Box(\Gamma\vd T)$.
    Then it must be the case that $\Sigma';u:\Gamma\vd T;\cdot\vd _1 M'': T_1$.
    By typing rule (Tm-GVar), all occurrences of $u$ in $M''$ must be associated with a substitution $\delta_k$
    that closes all free variables in $\Gamma$.
    In other words, all occurrences of $u$ in the term $M''$ can be listed as
    $$\; \cdots\; u^{\delta_1} \; \cdots \; u^{\delta_2}\; \cdots \; u^{\delta_3}\;\cdots$$
    where $\spccv{i}\delta_k:\Gamma$ for all $k$.
    Then $M''[M_1^\sigma/u]$ can be rewritten as
    $$\; \cdots\; M_1^\sigma[\delta_1] \; \cdots \; M_1^\sigma[\delta_2]\; \cdots \; M_1^\sigma[\delta_3]\;\cdots$$
    Let $N_1$ be the term
    $$\; \cdots\; (\la y.M_1^\sigma[\delta_1])\ () \; \cdots \; M_1^\sigma[\delta_2]\; \cdots \; M_1^\sigma[\delta_3]\;\cdots$$
    By \Cref{lem: beta preserve}, we have $(N_1,h)\Da$.
    Let $N_2$ be the term
    $$\; \cdots\; w\ () \; \cdots \; M_1^\sigma[\delta_2]\; \cdots \; M_1^\sigma[\delta_3]\;\cdots$$
    Then $((\la w.N_2)(\la y.M_1^\sigma[\delta_1]), h)\Da$.
    By \Cref{lem: CIU abstraction}, we have $\la y.M_1\lsmc \la y.M_2$.
    Therefore, $((\la w.N_2)(\la y.M_2^\sigma[\delta_1]), h)\Da$.
    Let $N_3$ be the term
    $$\; \cdots\; (\la y.M_2^\sigma[\delta_1])\ () \; \cdots \; M_1^\sigma[\delta_2]\; \cdots \; M_1^\sigma[\delta_3]\;\cdots$$
    Then $(N_3,h)\Da$.
    By \Cref{lem: beta preserve} again,
    we have $(N_4,h)\Da$ where $N_4$ is the term
    $$\; \cdots\; M_2^\sigma[\delta_1] \; \cdots \; M_1^\sigma[\delta_2]\; \cdots \; M_1^\sigma[\delta_3]\;\cdots$$
    Repeat the above process for all occurrences of $u$,
    we have $(N_5,h)\Da$ where $N_5$ is the term
    $$\; \cdots\; M_2^\sigma[\delta_1] \; \cdots \; M_2^\sigma[\delta_2]\; \cdots \; M_2^\sigma[\delta_3]\;\cdots$$
    This is exactly $(M''[M_2^\sigma/u]],h)\Da$.
\end{proof}

\begin{lemma}[CIU Precongruence]\label{lem: CIU precongruence}
    Given $\spgv{i} M_1,M_2:T$, if $M_1\lsmc M_2$,
    then $C[M_1]\lsmc C[M_2]$ for any context $C$.
\end{lemma}
\begin{proof}
    By induction on $C$,
    using \Cref{lem: CIU basics,lem: CIU box}.
\end{proof}

\begin{theorem}[CIU, \Cref{thm: CIU}]
    Given two terms $\spgv{i} M_1, M_2:T$, we have $M_1\lsm M_2$ if and only if $M_1\lsmc M_2$.
\end{theorem}
\begin{proof}
    The left-to-right direction holds because, given any $K$, $\sigma$, $\delta$, and $h$,
    we can always find $C$ such that
    $(C[M],\cdot)\ra^* (K[M[\sigma][\delta]],h)$ for any $M$.
    The right-to-left direction holds by \Cref{lem: CIU precongruence}.
\end{proof}

\subsection{Proofs for \Cref{sec: full abstraction}}

Suitably well-behaved pairs of passive and active configurations turn out to give rise to a pair of term and heap
by unfolding all the substitutions and evaluation contexts and merging heaps.
Below we make this intuition precise by defining the \emph{composition} of two \emph{compatible} configurations.

Given a mapping $\gamma$ from names to values,
we define substitutions $\delta^i$ by:
$\delta^0(n)=\gamma(n)$ and $\delta^{i+1}(n)=\delta^i(n)[\gamma]$.
If there exists $k$ such that $\delta^k=\delta^{k+1}$,
we define the \emph{idempotent substitution} $\idm{\gamma}$ to be $\delta^k$.

Given two stacks of evaluation contexts $\xi_1$ and $\xi_2$,
we define their \textit{unfolding} $\mergeConf{\xi_1}{\xi_2}$ to be:
(i) $\Sigma\vd K:(\cdot;\cdot;T_1;i_1)\ra (\cdot;\cdot;T_2;i_2)$
if $\xi_1=\cdot:(\Sigma\vd K:(\cdot;\cdot;T_1;i_1)\ra (\cdot;\cdot;T_2;i_2))$ and $\xi_2=\cdot$;
(ii) $\Sigma_1\cup\Sigma_2\vd K_2[K_1]:(\cdot;\cdot;T_1;i_1)\ra(\cdot;\cdot;T_3;i_3)$
if $\xi_1=\xi_3:(\Sigma_1\vd K_1:(\cdot;\cdot;T_1;i_1)\ra(\cdot;\cdot;T_2;i_2))$,
and $\mergeConf{\xi_2}{\xi_3}=(\Sigma_2\vd K_2:(\cdot;\cdot;T_2;i_2)\ra(\cdot;\cdot;T_3;i_3))$;
(iii) undefined otherwise.

\begin{definition}
A passive configuration $\CC_1=(\Sigma_1,h_1,\gamma_1,\phi_1,\xi_1,\theta_1)$
and an active configuration $\CC_2=(M,i,T,\Sigma_2,h_2,\gamma_2,\phi_2,\xi_2,\theta_2)$
are \textit{compatible} if $\dom{h_1}\cap\dom{h_2}=\theta_1=\theta_2$,
$\dom{\gamma_1}\uplus\dom{\gamma_2}=\phi_1=\phi_2$,
$\iota(\gamma_1\uplus\gamma_2)$ is well defined,
and $\mergeConf{\xi_1}{\xi_2}=(\Sigma_1\cup\Sigma_2 \vd K:(\cdot;\cdot;T;i)\ra (\cdot;\cdot;T';i'))$.
If $\CC_1$ and $\CC_2$ are compatible,
we define their \textit{composition} as $\mergeConf{\CC_1}{\CC_2}=(K[M], h_2+h_1)[\iota(\gamma_1\uplus\gamma_2)]$.
\end{definition}

The lemmas below display dynamic invariance properties stating that transitions between configurations
correspond to reductions between terms.
First we show that $P\tau$ transitions within the active configuration
correspond to reductions in the composed term and heap.

\begin{lemma}\label{lem:invariance_tau}
    Suppose passive configuration $\CC_1=(\Sigma_1,h_1,\gamma_1,\phi_1,\xi_1,\theta_1)$
    and active configuration $\CC_2=(M,i,T,\Sigma_2,h_2,\gamma_2,\phi_2,\xi_2,\theta_2)$
    are compatible.
    Also suppose $\CC_2\xra{\tau}^*\CC_2'$ in $k$ steps.
    Then $\mergeConf{\CC_1}{\CC_2}\ra^*\mergeConf{\CC_1}{\CC_2'}$ in $k$ steps.
\end{lemma}

\begin{proof}
    We prove by induction on $k$.
    If $k=0$, then $\CC_2=\CC_2'$, and the conclusion is trivial.
    Otherwise, suppose $\CC_2\xra{\tau}\CC_2''\xra{\tau}^*\CC_2'$.
    Then $\CC_2''=(M'',i,T,\Sigma_2'',h_2'',\gamma_2,\phi_2,\xi_2,\theta_2)$ where $(M,h_2)\ra(M'',h_2'')$.
    Therefore, we have
    \begin{equation*}\begin{split}
        & \mergeConf{\CC_1}{\CC_2}=(\mergeConf{\xi_1}{\xi_2})[M], h_2+h_1)[\iota(\gamma_1\uplus\gamma_2)]\\
        \ra &\mergeConf{\CC_1}{\CC_2''}=(\mergeConf{\xi_1}{\xi_2})[M''],h_2''+h_1)[\iota(\gamma_1\uplus\gamma_2)] \\
    \end{split}\end{equation*}
    Notice in particular that the above reduction is also valid
    under the special reduction rule (Letbox name).
    As $\CC_2''\xra{\tau}^*\CC_2'$ in $k-1$ steps,
    by induction hypothesis,
    $\mergeConf{\CC_1}{\CC_2''}\ra^*\mergeConf{\CC_1}{\CC_2'}$ in $k-1$ steps.
    Therefore, $\mergeConf{\CC_1}{\CC_2}\ra^*\mergeConf{\CC_1}{\CC_2'}$ in $k$ steps.
\end{proof}

\begin{example}
    The preceding lemma is where the LTS for closed types failed. Consider the term
    $$\ag{x}$$
    and the evaluation context
    $$(\run \ag{\la x.\es \bullet})1$$
    After the P action $\ol{b_1}$, the O configuration becomes active and evaluates as shown below:
    $$(\run \ag{\la x.\es b_1})1\ra (\run \ag{\la x.\# b_1})1\ra (\la x.\# b_1)1\ra \# b_1$$
    Correspondingly, the composed term should evaluate as follows:
    $$(\run \ag{\la x.\es \ag{x}})1\ra (\run \ag{\la x.x})1\ra (\la x.x)1\ra 1$$
    Notice the mismatch in the last step because the substitution $[1/x]$ is forgotten.
    For contextual types, the O configuration will evaluate as:
    $$(\run \ag{\la x.\es b_1})1\ra (\run \ag{\la x.\# b_1})1\ra (\la x.\# b_1)1\ra \# b_1^{1/x}$$
    So this correspondence between configuration and composed term is preserved.
\end{example}

Suppose $\tr\ =\ \act_1\;\act_2\dots\act_n$.
Then we write $\CC\xra{\tr}\CC'$ if $\CC\xra{\act_1}\CC''\xra{\act_2}\cdots\xra{\act_n}\CC'$,
and $\CC\xrsa{\tr}\CC'$ if $\CC\xra{\tr}\CC''\xra{\tau}\CC'$.
In the lemma below, we show that transitions involving only actions
do not change the composed term and heap.

\begin{lemma}\label{lem:invariance_act}
    Suppose passive configuration $\CC_1=(\Sigma_1,h_1,\gamma_1,\phi_1,\xi_1,\theta_1)$
    and active configuration $\CC_2=(M,i,T,\Sigma_2,h_2,\gamma_2,\phi_2,\xi_2,\theta_2)$
    are compatible.
    Also suppose $\CC_2\xra{\act}\CC_2'$ and $\CC_1\xra{\ol{\act}}\CC_1'$.
    Then $\mergeConf{\CC_1}{\CC_2}=\mergeConf{\CC_2'}{\CC_1'}$.
\end{lemma}

\begin{proof}
    We prove by case analysis on $\act$.
    \begin{itemize}
        \item (PA). Then $\CC_2\xra{(\ol{A},\Sigma_3,\chi)}\CC_2'$ and $\CC_1\xra{(A,\Sigma_3,\chi)}\CC_1'$, where
        \begin{itemize}
            \item $\CC_2=(V,i,T,\Sigma_2,h_2,\gamma_2,\phi_2,\xi_2,\theta_2)$
            \item $\CC_2'=(\Sigma_2,h_2,\gamma_2\uplus\gamma_2'\uplus\gamma_2'',\phi_2',\xi_2,\dom{\chi})$
            \item $\CC_1=(\Sigma_1,h_1,\gamma_1,\phi_1,\xi_1':(\Sigma_1'\vd K:(\cdot;\cdot;T;i)\ra (\cdot;\cdot;T';i')),\theta_1)$
            \item $\CC_1'=(K[A],T',i',\Sigma_1\cup\Sigma_3,\chi+h_1,\gamma_1,\phi_1',\xi_1',\dom{\chi})$
            \item $(A,\gamma_2')\in\AVal{\sccv{i} V:T}$
            \item $(\chi,\gamma_2'')\in\AVal{\sccv{} \rst{h_2}{\sigma\cup\mu(A)}:\heap}$
        \end{itemize}
        Therefore,
        \begin{equation*}\begin{split}
            & \mergeConf{\CC_1}{\CC_2}=((\mergeConf{\xi_1}{\xi_2})[V], h_2+h_1)\bc{\iota(\gamma_1\uplus\gamma_2)}\\
            =&\mergeConf{\CC_2'}{\CC_1'}=((\mergeConf{\xi_2}{\xi_1'})[K[A]],(\chi+h_1)+h_2)\bc{\delta(\gamma_1\uplus(\gamma_2\uplus\gamma_2'\uplus\gamma_2''))}
        \end{split}\end{equation*}
        \item (PQF). Then $\CC_2\xra{(\ol{f}(A),\Sigma_3,\chi)}\CC_2'$ and $\CC_1\xra{(f(A),\Sigma_3,\chi)}\CC_1'$, where
        \begin{itemize}
            \item $\CC_2=(K[fV],i,T,\Sigma_2,h_2,\gamma_2,\phi_2,\xi_2,\theta_2)$
            \item $\CC_2'=(\Sigma_2,h_2,\gamma_2\uplus\gamma_2'\uplus\gamma_2'',\phi_2',\xi_2:(\Sigma_2\vd K:(\cdot;\cdot;T_2;i_3)\ra(\cdot;\cdot;T;i)),\dom{\chi})$
            \item $\CC_1=(\Sigma_1,h_1,\gamma_1,\phi_1,\xi_1,\sigma_1)$
            \item $\CC_1'=(\gamma_1(f)A,i_3,T_2,\chi+h_1,\gamma_1,\phi_1',\xi_1,\dom{\chi})$
            \item $f\in\FNames_{T_1\ra T_2, i_1}$
            \item $(A,\gamma_2')\in\AVal{\Sigma_2;\cdot;\cdot\vd_{i_2}V:T_1}$
            \item $i_3=\max(i_1,i_2)$
            \item $(\chi,\gamma_2'')\in\AVal{\sccv{}\rst{h_2}{\sigma\cup\mu(A)}:\heap}$
        \end{itemize}
        Therefore,
        \begin{equation*}\begin{split}
        & \mergeConf{\CC_1}{\CC_2}=((\mergeConf{\xi_1}{\xi_2})[K[fV]], h_2+h_1)\{\iota(\gamma_1\uplus\gamma_2)\}\\
        =& \mergeConf{\CC_2'}{\CC_1'}=((\mergeConf{(\xi_2:(\Sigma_2\vd K:(\cdot;\cdot;T_2;i_3)\ra(\cdot;\cdot;T;i)))}{\xi_1})[\gamma_1(f)A],(\chi+h_1)+h_2)\bc{\delta(\gamma_1\uplus\gamma_2\uplus\gamma_2'\uplus\gamma_2'')}
        \end{split}\end{equation*}
    \item The other cases are similar therefore omitted.
    \end{itemize}
\end{proof}

Now we show the other direction of the invariance, i.e.,
that reductions in the composed term correspond to transitions in configurations.
Recall that the trace $\ol{\tr}$ is obtained by reverting the polarity of each action in $\tr$, i.e., P actions become O actions and O actions become P actions.

Given term $M$ and name substitution $\gamma$,
we define the \textit{depth} of $M$ under $\gamma$,
denoted by $\depth{\gamma}{M}$, as

$$\depth{\gamma}{M}=\sum _{\text{occurence of }b\in\BNames\text{ in }M}\left(\depth{\gamma}{\gamma(b)}+1\right)$$

Notice that the above summation is for \textit{occurences} of $\BNames$.
That is, if $b$ occurs $k$ times in $M$,
then it should be counted $k$ times in the summation.
If the idempotent substitution $\delta(\gamma)$ exists,
then the function $\depth{\gamma}{-}$ is well defined for all terms that contain only names in $\dom{\gamma}$,
as $\gamma$ gives a partial order on names.
Also, if $\gamma\subseteq\gamma'$,
then for all such terms, $\depth{\gamma}{M}=\depth{\gamma'}{M'}$.
Therefore, the depth of a term remains unchanged in future transitions,
and we can simply write it as $\depth{}{M}$.

Intuitively, $\depth{\gamma}{M}$ measures how many unfoldings of $\BNames$ are needed
to obtain a term without any $\BNames$.
It is a curcial invariant in the lemma below,
which will then be used to show that there cannot be infinite sequences of $\tau$-free transitions.

Below we will frequently encounter situations where
the polarity of an action does not matter.
So we will use (A) to refer to both (PA) and (OA) actions,
and (QF) to refer to both (PQF) and (OQF) actions,
and (QB) to refer to both (PQB) and (OQB) actions.

\begin{lemma}\label{lem: depth bound}
    Suppose passive configuration $\CC_1=(\Sigma_1,h_1,\gamma_1,\phi_1,\xi_1,\theta_1)$
    and active configuration $\CC_2=(M,i,T,\Sigma_2,h_2,\gamma_2,\phi_2,\xi_2,\theta_2)$
    are compatible.
    Also suppose $\CC_1\xra{\tr}\CC_1'$, $\CC_2\xra{\ol{\tr}}\CC_2'$,
    where $\tr$ consists only of (A) and (QB) actions,
    and any prefix of $\tr$ contains at least as many (QB) actions
    as (A) actions.
    Then the number of (QB) actions in $\tr$ is at most $\depth{\gamma_1\uplus\gamma_2}{M}$.
\end{lemma}

\begin{proof}
    We prove by induction on $\depth{\gamma_1\uplus\gamma_2}{M}$.
    If $\depth{}{M}=0$,
    then $M$ does not contain any $\BNames$,
    so it cannot yield any (QB) actions,
    and $\tr$ can only be empty.
    Below we suppose result for all smaller values of depth,
    and prove for $\depth{}{M}=d>0$.

    For $M$ to yield (PQB) action in the first step,
    it must be of the form $K[\#b^\delta]$,
    where $b\in\BNames_{\Box(\Gamma\vd T'), 1}$ and $\Sigma_2;\cdot;\cdot\vd\delta:\Gamma$.
    In such case, we have

    $$\depth{}{M}=\depth{}{K} + \depth{}{\gamma_1(b)}+ 1$$

    and $\CC_2\xra{(\ol{\run b^\rho}, \Sigma_3,\chi_1)}\CC_2''$,
    $\CC_1\xra{(\run b^\rho,\Sigma_3,\chi_1)}\CC_1''$, where
    $$\CC_1''=(\#\gamma_1(b)[\rho],i_3,T_2,\Sigma_3\cup\Sigma_1,\chi_1+h_1,\gamma_1,\phi_1'',\xi_1,\dom{\chi})$$
    Since $\vd_1 \Box(\Gamma\vd T'):\typ$, by (Typ-Box), we have $\vd_0 \Gamma:\lctx$.
    Then for all $x\in\dom{\rho}$, $\Sigma_3\cup\Sigma_1;\cdot;\cdot\vd_0 \rho(x):\Gamma(x)$.
    Therefore, $\rho(x)$ cannot be a BName,
    and $\depth{}{\#\gamma_1(b)[\rho]}=\depth{}{\#\gamma_1(b)}=\depth{}{\gamma_1(b)}$.

    Recall that we supposed $\CC_1\xra{\tr}\CC_1'$, $\CC_2\xra{\ol{\tr}}\CC_2'$.
    So $\CC_1'\xra{\tr'}\CC_1''$ and $\CC_2'\xra{\ol{\tr'}}\CC_2''$,
    where $\tr = (\run b^\rho,\Sigma_3,\chi)\ \tr'$
    and $\tr'$ consists only of (A) and (QB) actions.
    If any prefix of $\tr'$ contains at least as many (QB) actions
    as (A) actions, then by induction hypothesis we know that the number of (QB) actions in $\tr'$
    is at most $\depth{}{\#\gamma_1(b)[\rho]}$.
    Therefore, the number of (QB) actions in $\tr$ is at most
    $\depth{}{\#\gamma_1(b)[\rho]}+1\le\depth{}{M}$.

    Otherwise, we locate the first prefix of $\tr'$
    where the number of (A) actions exceed the number of (QB) actions.
    Namely, we suppose $\CC_1'\xra{\tr''}\CC_1'''\xra{(\ol{A},\Sigma_4,\chi_2)}\CC_1''''\xra{\tr'''}\CC_1'$, 
    and $\CC_2'\xra{\ol{\tr''}}\CC_2'''\xra{(A,\Sigma_4,\chi_2)}\CC_2''''\xra{\ol{\tr'''}}\CC_2'$,
    where
    $$\tr=(\run b^\rho,\Sigma_3,\chi_1)\ \tr''\ (\ol{A},\Sigma_4,\chi_2)\ \tr'''$$
    and $\tr''$ contains an equal number of (QB) and (A) actions.
    So the action $(A,\Sigma_4,\chi_2)$ must be the answer to the initial question $(\run b^\rho,\Sigma_3,\chi_1)$,
    and we have
    $$\CC_2''''=(K[A],i,T,\Sigma_4,h_3,\gamma_3,\phi_3,\xi_3,\theta_3)$$
    As $\Sigma_4;\cdot;\cdot\vd_0 A:T'$, $A$ cannot be a BName,
    so $\depth{}{K[A]}=\depth{}{K}$.
    Also, the trace $\CC_1\ra{(\run b^\rho,\Sigma_3,\chi_1)\ \tr''\ (\ol{A},\Sigma_4,\chi_2)}\CC_1''''$
    contains an equal number of (QB) and (A) actions,
    so any prefix of $\tr'''$ contains at least as many (QB) actions as (A) actions.
    By the induction hypothesis, the number of (QB) actions in $\tr'''$
    is at most $\depth{}{K[A]}$.
    Therefore, the number of (QB) actions in $\tr$ is at most
    $1+\depth{}{\#\gamma_1(b)[\rho]}+\depth{}{K[A]}=\depth{}{M}$.

\end{proof}

\begin{lemma}\label{lem:finite trace}
    Suppose passive configuration $\CC_1=(\Sigma_1,h_1,\gamma_1,\phi_1,\xi_1,\theta_1)$
    and active configuration $\CC_2=(M,i,T,\Sigma_2,h_2,\gamma_2,\phi_2,\xi_2,\theta_2)$
    are compatible.
    Then there exists finite trace $\tr$ such that $\CC_1\xra{\tr}\CC_1'$, $\CC_2\xra{\ol{\tr}}\CC_2'$,
    and either (a) $\CC_1'$ (or $\CC_2'$) is an active configuration whose next transition is a $\tau$ transition,
    or (b) $\CC_1'$ (or $\CC_2'$) is an active configuration whose next transition is a (PA) transition,
    and $\CC_2'$ (or $\CC_1'$, resp.) is a passive configuration whose stack $\xi$ is empty.
    In other words, there does not exist infinite sequences of $\tau$-free transitions.
\end{lemma}
\begin{proof}
    We prove by case analysis on first step of transition the active configuration $\CC_2$.
    \begin{itemize}
        \item (PQF). Then $\CC_2\xra{(\ol{f}(A), \Sigma_3,\chi)}\CC_2'$,
        and $\CC_1\xra{(f(A),\Sigma_3,\chi)}\CC_1'$, where
        \begin{itemize}
            \item $\CC_2=(K[fV],i,T,\Sigma_2,h_2,\gamma_2,\phi_2,\xi_2,\theta_2)$
            \item $\CC_2'=(\Sigma_2,h_2,\gamma_2',\phi_2',\xi_2:(\Sigma_2\vd K:(\cdot;\cdot;T_2;i_3)\ra(\cdot;\cdot;T;i)),\dom{\chi})$
            \item $\CC_1=(\Sigma_1,h_1,\gamma_1,\phi_1,\xi_1,\sigma_1)$
            \item $\CC_1'=(\gamma_1(f)A,i_3,T_2,\chi+h_1,\gamma_1,\phi_1',\xi_1,\dom{\chi})$
            \item $(A,\gamma_2')\in\AVal{\Sigma_2;\cdot;\cdot\vd_{i_2}V:T_1}$
        \end{itemize}
        The term in the new active configuration $\CC_1'$ is $\gamma_1(f)A$,
        where $\gamma_1(f)$ is either $\la x.M'$ or $f'$ for some $f'\in\FNames_{T_1\ra T_2, i_1}$.
        In the former case, the next transition of $\CC_1'$ is a $\tau$ transition $((\la x.M')A,\chi+h_1)\ra(M'[A/x],\chi+h_1)$,
        and clause (a) holds.
        In the latter case, the next transition of $\CC_1'$ is a (PQF) transition $(\ol{f'}(A'), \Sigma_3,\chi')$,
        This process will eventually stop as $\gamma$ gives a partial order on names,
        and we will end up in the first case.
        \item (PQB). Then $\CC_2\xra{(\ol{\run f^\rho}, \Sigma_3,\chi)}\CC_2'$,
        and $\CC_1\xra{(\run f^\rho,\Sigma_3,\chi)}\CC_1'$, where
        \begin{itemize}
            \item $\CC_2=(K[\#b^\delta],i,T,\Sigma_2,h_2,\gamma_2,\phi_2,\xi_2,\theta_2)$
            \item $\CC_2'=(\Sigma_2,h_2,\gamma_2',\phi_2',\xi_2:(\Sigma_2\vd K:(\cdot;\cdot;T_2;i_3)\ra(\cdot;\cdot;T;i)),\dom{\chi})$
            \item $\CC_1=(\Sigma_1,h_1,\gamma_1,\phi_1,\xi_1,\sigma_1)$
            \item $\CC_1'=(\#\gamma_1(b)[\rho],i_3,T_2,\chi+h_1,\gamma_1,\phi_1',\xi_1,\dom{\chi})$
        \end{itemize}
        We further consider the next transition of $\CC_1'$ based on what $\gamma_1(b)$ is.
        \begin{itemize}
            \item If $\gamma_1(b)=\bx K[M]$, where $M$ is a redex,
            then the next transition is a $\tau$ transition, and clause (a) is satisfied.
            \item If $\gamma_1(b)=\bx K[fV]$, then the next transition is a (PQF) action,
            and we go to the previous case.
            \item If $\gamma_1(b)=b'$, then the next transition is a (PQB) action $(\ol{\run b'^{\rho'}}, \Sigma_3,\chi')$.
            \item If $\gamma_1(b)=\bx K[\#b'^{\delta'}]$,
            then the next transition is also a (PQB) action $(\ol{\run b'^{\rho'}}, \Sigma_3,\chi')$.
            \item If $\gamma_1(b)=\bx V$, then the next transition is a (PA) action $(\ol{A}, \Sigma_3,\chi')$.
        \end{itemize}
        Now we still need to argue that there does not exist an infinite trace consisting only of (QB) and (A) actions.
        Suppose there exists a (possibly infinite) trace $\CC_1\xra{a_1}\CC_1^1\xra{a_2}\CC_1^2\xra{a_3}\cdots$
        and $\CC_2\xra{\ol{a_1}}\CC_2^1\xra{\ol{a_2}}\CC_2^2\xra{\ol{a_3}}\cdots$.
        Suppose the evaluation context stack in $\CC_1^i$ and $\CC_2^i$ are $\xi_1^i$ and $\xi_2^i$ respectively,
        and we define their sizes as $\size{\cdot}=0$, $\size{\xi:K}=\size{\xi}+1$.
        We find $i$ such that $\size{\xi_1^i}+\size{\xi_2^i}$ is the smallest.
        If there are multiple such $i$s, we choose the smallest one.
        As a (QB) action increases $\size{\xi_1^i}+\size{\xi_2^i}$ by one,
        and an (A) action decreases $\size{\xi_1^i}+\size{\xi_2^i}$ by one,
        we know that the trace starting from $\CC_1^i$ and $\CC_2^i$ always contains
        at least as many (QB) actions as (A) actions.
        By \Cref{lem:finite trace}, this trace contains at most $\depth{}{M^i}$ (QB) actions,
        where $M^i$ is the term in the active configuration.
        Therefore, this trace cannot be infinite.

        \item (PA). There cannot be an infinite sequence of (PA) actions
        because each (PA) action consumes an evaluation context on the stack.
        So eventually either the stack becomes empty, in which case clause (b) holds,
        or we have a non-(PA) transition, in which case we go to previous discussions.
    \end{itemize}
\end{proof}

\begin{lemma}\label{lem:progression_act}
    Suppose passive configuration $\CC_1=(\Sigma_1,h_1,\gamma_1,\phi_1,\xi_1,\theta_1)$
    and active configuration $\CC_2=(M,i,T,\Sigma_2,h_2,\gamma_2,\phi_2,\xi_2,\theta_2)$
    are compatible,
    and $\mergeConf{\CC_1}{\CC_2}\ra (M',h')$.
    Then there exists $\CC_1'$, $\CC_2'$ such that
    either (a) $\CC_1\xra{\tr}\CC_1'$, $\CC_2\xrsa{\ol{\tr}}\CC_2'$,
    and $\mergeConf{\CC_1'}{\CC_2'}=(M',h')$,
    or (b) $\CC_1\xrsa{\tr}\CC_1'$, $\CC_2\xra{\ol{\tr}}\CC_2'$,
    and $\mergeConf{\CC_2'}{\CC_1'}=(M',h')$.
\end{lemma}

\begin{proof}
    By \Cref{lem:finite trace},
    there exists finite trace $\tr$ such that $\CC_1\xra{\tr}\CC_1''$, $\CC_2\xra{\ol{\tr}}\CC_2''$,
    and either (a) $\CC_1''$ (or $\CC_2''$) is an active configuration whose next transition is a $\tau$ transition,
    or (b) $\CC_1''$ (or $\CC_2''$) is an active configuration whose next transition is a (PA) transition,
    and $\CC_2''$ (or $\CC_1''$, resp.) is a passive configuration whose stack $\xi$ is empty.
    We discuss these two cases separately.
    \begin{itemize}
        \item Case (a). Without loss of generality, suppose $\CC_1''$ is the active configuration,
        and $\CC_1''\xra{\tau}\CC_1'''$.
        By \Cref{lem:invariance_act},
        $\mergeConf{\CC_1}{\CC_2}=\mergeConf{\CC_1''}{\CC_2''}$.
        By \Cref{lem:invariance_tau},
        $\mergeConf{\CC_1''}{\CC_2''}\ra\mergeConf{\CC_1'''}{\CC_2''}$.
        Therefore, $\CC_1'=\CC_1'''$ and $\CC_2'=\CC_2''$ satisfy the requirement.
        \item Case (b). Without loss of generality, suppose $\CC_1''$ is the active configuration,
        the term in $\CC_1''$ is a value $V$,
        and the stack $\xi_2''$ in $\CC_2''$ is empty.
        Then $\mergeConf{\CC_1}{\CC_2}=\mergeConf{\CC_1''}{\CC_2''}$ is also value,
        but we know that $\mergeConf{\CC_1}{\CC_2}\ra (M',h')$,
        so this case is not possible.
    \end{itemize}
\end{proof}

Reductions of compositions turn out to correspond to synchronizing their actions.
\begin{definition}
Given a sequence $\tr$ of actions, let us write $\ol{\tr}$ for the sequence obtained by converting 
each action to one of opposite polarity (i.e. context to program, program to context).
We write $\CC_1|\CC_2\da^\tr$ if either (a) $\tr\in\Tr{\CC_1}$ and $\ol{\tr}\; \act\in\mathbf{Tr}(\CC_2)$, 
or (b) $\tr\; \act\in\Tr{\CC_1}$ and $\ol{\tr}\in\mathbf{Tr}(\CC_2)$ for some (PA) action $\act$.
\end{definition}

\begin{lemma}\label{lem:correctness_inductive}
    Suppose passive $\CC_1$ and
    and active $\CC_2$ are compatible.
    Then $\mergeConf{\CC_1}{\CC_2}\Da^0$ iff $\CC_1\mid\CC_2\da^\tr$ for some $\tr$.
\end{lemma}

\begin{proof}
    First we prove if $\CC_1|\CC_2\da^\tr$ then $\mergeConf{\CC_1}{\CC_2}\Da^0$
    by induction on the length of $\tr$.
    If $\tr$ is empty, then $\xi_1=\xi_2=\cdot$, and $M$ is a value.
    Therefore, $\mergeConf{\CC_1}{\CC_2}$ is also a value.
    Otherwise, $\tr = \act\; \tr'$.
    Suppose $\CC_1\xRa{\act}\CC_1'$ and $\CC_2\xRa{\ol{\act}}\CC_2'$.
    Then $\CC_2'\mid\CC_1'\da^{\tr'}$.
    By induction bypothesis, $\mergeConf{\CC_2'}{\CC_1'}\Da^0$.
    Suppose $\CC_2\xra{\tau}^*\CC_2''\xra{\act}\CC_2'$,
    and, obviously, $\CC_1\xra{\ol{\act}}\CC_1'$.
    By \Cref{lem:invariance_tau}, $\mergeConf{\CC_1}{\CC_2}\ra^*\mergeConf{\CC_1}{\CC_2''}$.
    By \Cref{lem:invariance_act}, $\mergeConf{\CC_1}{\CC_2''}=\mergeConf{\CC_2'}{\CC_1'}$.
    Therefore, $\mergeConf{\CC_1}{\CC_2}\Da^0$.
    
    Then we prove if $\mergeConf{\CC_1}{\CC_2}\Da^0$ then $\CC_1|\CC_2\da^\tr$
    by induction on the number of steps in the reduction $\mergeConf{\CC_1}{\CC_2}\ra^*(V,h')$.
    If there are 0 steps, then $M$ is a value,
    so empty trace is the $\tr$ we are looking for.
    Otherwise, $\mergeConf{\CC_1}{\CC_2}\ra(M'',h'')\ra^*(V,h')$.
    By \Cref{lem:progression_act}, without loss of generality,
    there exists $\CC_1'$, $\CC_2'$ such that
    $\CC_1\xra{\tr'}\CC_1'$, $\CC_2\xrsa{\ol{\tr'}}\CC_2'$,
    and $\mergeConf{\CC_1}{\CC_2}\ra\mergeConf{\CC_1'}{\CC_2'}$.
    Then $\mergeConf{\CC_1'}{\CC_2'}\Da^0$ as well,
    and by induction hypothesis we have $\CC_1|\CC_2\da^{\tr''}$.
    Therefore, $\CC_1|\CC_2\da^{\tr'\; \tr''}$.

\end{proof}

The above result can be instantiated to initial configurations for programs and contexts to yield the following theorem.
\begin{theorem}[Correctness, \Cref{thm:correctness}]
    For any term $\spgv{i}M:T$,
    evaluation context $\Sigma'\vd K:(\cdot;\cdot;i;T)\ra (\cdot;\cdot;i';T')$,
    global substitution $\Sigma';\cdot\vd\sigma:\Psi$,
    local substitution $\spccv{i}\delta:\Gamma$,
    and heap $\spccv{} h:\heap$
    such that $\Sigma'\supseteq\Sigma$,
    let $(\eta,\gamma)\in\AVal{\Sigma';\cdot\vd\sigma:\Psi}$,
    $(\rho,\gamma')\in\AVal{\spccv{i}\delta:\Gamma}$,
    $(\Sigma'',h')=\rst{(\Sigma',h)}{\mu(\eta)\cup\mu(\rho)}$,
    and $(\chi,\gamma'')\in\AVal{\Sigma'';\cdot;\cdot\vd h':\heap}$,
    then $(K[M[\sigma][\delta]],h)\Da$ iff
    there exists trace $\tr$ such that
    $\cconf{K,\sigma,\delta,h}{\eta,\gamma,\rho,\gamma',\chi,\gamma''}\mid\cconf{M}{\eta,\rho,\chi}\da^{\tr}$.
\end{theorem}
\begin{proof}
    Follows from \Cref{lem:correctness_inductive} by noting that
    $\mergeConf{\cconf{K,\sigma,\delta,h}{\eta,\gamma,\rho,\gamma',\chi,\gamma''}}{\cconf{M}{\eta,\rho,\chi}}=(K[M[\sigma][\delta]],h)$.
\end{proof}
It follows that the trace semantics of terms from \Cref{dfn:semantics} is sound for $\lsm$.
\begin{theorem}[Soundness, \Cref{thm: soundness}]
    Given two terms $\spgv{i} M_1,M_2:T$,
    if $\Tr{M_1}\subseteq\Tr{M_2}$,
    then $M_1\lsm M_2$.
\end{theorem}

\begin{proof}
    Suppose there exists a global substitution $\Sigma';\cdot\vd \sigma:\Psi$,
    local substitution $\spccv{i} \delta:\Gamma$,
    evaluation context $\spccv{} K: (\cdot;\cdot;i;T)\ra(\cdot;\cdot;i';T')$,
    and heap $\spccv{} h:\heap$
    such that $\Sigma'\supseteq\Sigma$ and $(K[M_1[\sigma][\delta]],h)\Da^0$.
    Also suppose $(\eta,\gamma)\in\AVal{\Sigma';\cdot\vd\sigma:\Psi}$,
    $(\rho,\gamma')\in\AVal{\spccv{i}\delta:\Gamma}$,
    $(\Sigma'',h')=\rst{(\Sigma',h)}{\mu(\eta)\cup\mu(\rho)}$,
    and $(\chi,\gamma'')\in\AVal{\Sigma'';\cdot;\cdot\vd h':\heap}$.
    By \Cref{thm:correctness}, there exists $\tr$ and $\act$
    such that $\tr\in\Tr{\cconf{M_1}{\eta,\rho,\chi}}$ and $\ol{\tr}\; \act\in\Tr{\cconf{K,\sigma,\delta,h}{\eta,\gamma,\rho,\gamma',\chi,\gamma''}}$.
    Since $\Tr{M_1}\subseteq\Tr{M_2}$, we have $\tr\in\Tr{\cconf{M_2}{\eta,\rho,\chi}}$.
    By \Cref{thm:correctness} again, $(K[M_2[\sigma][\delta]],h)\Da^0$.
\end{proof}

To show completeness, we  prove a definability property stating that,
for any well-behaved sequence of actions (such as those produced the Program LTS),
we can find a configuration that generates this trace only.
To this end, we need to pin down what well-behaved means.
This will be captured by the definition of an $N_O,N_P$-trace, 
where $N_O,N_P\subseteq\Names$ are `initial sets' of names corresponding
to values of the context and the program respectively.
Crucially, the definition will capture various freshness properties and the stack discipline enforced in the LTS.
We shall say that a name is \textit{introduced} in an action
if it is used as an answer, as the argument in a Question of Function action,
as the codomain of the substitution in a Question of Box action,
or appears in $\chi$. In other words, a name is introduced \textit{unless} it is
the $f$ in question $f(A)$ or the $b$ in question $\run b^\rho$.
\begin{definition}
Given $N_O,N_P\subseteq\Names$, an \textit{$N_O,N_P$-trace} is a sequence of actions such that
\begin{itemize}
    \item Player and Opponent actions alternate.
    \item No name in $N_P\cup N_O$ is introduced,
    and no name is introduced twice.
    \item If a name is used as a question,
    then it must either occur in the set $N_\_$ of opposite polarity or have been introduced in a previous action of opposite polarity.
    \item We say a Player Answer \textit{is an answer to} an Opponent Question
    if it is the first player answer such that the number of Player Answers between them equals the number of Opponent Questions between them.
    If $(\ol{A_2},\Sigma_2,\chi_2)$ is an answer to $(f(A_1),\Sigma_1,\chi_1)$,
    then $f\in \FNames_{T_1\ra T_2,i_1}$,
    and $\Sigma_2;\cdot;\cdot\vd_{i_2}A_2:T_2$.
    If $(\ol{A},\Sigma_2,\chi_2)$ is an answer to $(\run b^\rho,\Sigma_1,\chi_1)$,
    then $b\in \BNames_{\Box(\Gamma\vd T), 1}$ and $\Sigma_2;\cdot;\cdot\vd_i A:T$.
    The same requirement applies to Opponent Answers and Player Questions.
    \item Suppose $\Sigma,\chi$ and $\Sigma',\chi'$ belong to two neighbouring actions $\act$ and $\act'$.
    Then $\dom{\chi'}=\dom(\Sigma')$ should be the minimal set of locations satisfying the following three conditions:
    (a) $\dom{\chi}\subseteq\dom{\chi'}$,
    (b) locations occurring in $\act'$ (if any) occur in $\dom{\chi'}$,
    and (c) $\chi'(\ell)=\ell'$ implies $\ell'\in\dom{\chi'}$.
\end{itemize}
\end{definition}
We say an $\emptyset,N_P$-trace is \textit{complete} if
\begin{itemize}
    \item it starts with an Opponent action,
    \item it ends with a (PA) action,
    \item the number of (OA) actions equals the number of (PQ) actions plus one,
    \item in any prefix of it the number of (OA) actions is at most the number of (PQ) actions plus one,
    \item and in any proper prefix of it the number of (PA) actions is at most the number of (OQ) actions,
\end{itemize}

Then for $\tr\in\Tr{\CC}$, $\tr$ is a complete trace iff $\CC$ is a passive configuration
whose stack $\xi$ contains exactly one evaluation context.
In particular, if $\tr\in\Tr{\cconf{K,\sigma,\delta,h}{\eta,\gamma,\rho,\gamma',\chi,\gamma''}}$,
then $\tr$ must be a complete $\emptyset,\dom{\gamma}\uplus\dom{\gamma'}\uplus\dom{\gamma''}$-trace.
\begin{lemma}\label{lem:definability_inductive}
    Suppose $\tr\ =\ \act_1\dots\act_n$ is a complete $\emptyset,N_P$-trace.
    Let $\tr_i\ =\ \act_i\dots\act_n$
    and, in particular, let $\tr_{n+1}$ be the empty trace.
    Then there exists configurations $\CC_i$
    such that $\Tr{\CC_i}=\{\tr_i\}$ up to renaming that preserves $N_P$ and names that occur in $\act_1\dots\act_{i-1}$.
\end{lemma}
\begin{proof}[Proof Sketch]
    We enumerate all the names and PQ actions that occur in the trace,
    and create references for each of them.
    The references for names are set appropriately so that only the designated actions can be taken at each step,
    and all other actions will lead to divergence.
    The stacks are composed of evaluation contexts
    that put the terms inside corresponding PQ references.
    The terms $m_i$ will be specified by backward induction on $i$.
    For P actions, $m_i$ should set the locations according to $\chi_i$ and generate the corresponding player action.
    For O actions, $m_i$ should assert that the action is received as expected,
    assert that the shared locations are updated as expected,
    set the stored values to be the same as $h_{i+1}$,
    and return $m_{i+1}$.
\end{proof}
\begin{proof}
    The components of $\CC_i=(M_i,l_i,T_i,\Sigma_i,h_i,\gamma_i,\phi_i,\xi_i,\theta_i)$
    or $\CC_i=(\Sigma_i,h_i,\gamma_i,\phi_i,\xi_i,\theta_i)$ are specified as follows.
    For simplicity, some typing judgments are omitted if they are clear from the context.
    \begin{itemize}
        \item We enumerate function names that either occur in $N_P$
        or are introduced by P in $\tr$ as $f_{P,j}(1\le j\le n_{PF})$.
        We enumerate box names that either occur in $N_P$
        or are introduced by P in $\tr$ as $b_{P,j}(1\le j\le n_{PB})$.
        We create references $\pfr{j}(1\le j\le n_{PF})$ and $\pbr{j}(1\le j\le n_{PB})$
        of corresponding type for each of them
        so that we can control whether or not to enable an O action
        by manipulating the heap.

        Similarly, we enumerate function names introduced by O as $f_{O,j}(1\le j\le n_{OF})$
        and box names introduced by O as $b_{O,j}(1\le j\le n_{OB})$.
        We create references $\ofr{j}$ and $\obr{j}$ of corresponding types
        so that the names will not occur in the configuration before they are introduced by O,
        and to preserve $\alpha$-renaming for such names.

        Also suppose there are $n_{PQ}$ player questions,
        and we create references $\pqr{j}(1\le j\le n_{PQ})$ for each of them.
        We also create a special reference $\pqr{0}$ that represents the initial evaluation context.
        The types of $\pqr{j}$ will be known after we specify the other components later.

        Let $\dom{\Sigma_i}=\dom{h_i}=\bc{\pfr{j}\mid 1\le j\le n_{PF}}
        \uplus\bc{\pbr{j}\mid 1\le j\le n_{PB}}
        \uplus\bc{\ofr{j}\mid 1\le j\le n_{OF}}
        \uplus\bc{\obr{j}\mid 1\le j\le n_{OB}}
        \uplus\bc{\pqr{j}\mid 0\le j\le n_{PQ}}
        \uplus\bc{\ell\mid \ell\mathrm{\ appears\ in\ } \tr}$.
        In this proof we use $\ell$ to denote the locations that appear in $\tr$ only,
        not the special references introduced above to represent names.

        If $f_{O,j}$ is introduced at $\act_{i'}$,
        then $h_i(\ofr{j})=\la x.(!\ofr{j})x$ for $i\le i'$ and $h_i(\ofr{j})=f_{O,j}$ for $i>i'$.
        If $b_{O,j}$ is introduced at $\act_{i'}$,
        then $h_i(\obr{j})=\bx(\letbox{u}{!\obr{j}}{u^{id}})$ for $i\le i'$ and $h_i(\obr{j})=b_{O,j}$ for $i\ge i'$,
        where $id$ is the identity substitution for all variables in $\Gamma$ if $b_{O,j}\in\BNames_{\Box(\Gamma\vd T),1}$.
        The values of the rest of $h_i(-)$ will be specified later.
        
        \item $\dom{\gamma_i}$ consists of all names in $N_P$ and all names introduced by P in $\act_1\dots\act_{i-1}$.
        Let $\gamma_i(f_{P,j})=\la x.(!\pfr{j})x$ and $\gamma_i(b_{P,j})=\bx (\letbox{u}{!\pbr{j}}{u^{id}})$ whenever they are defined.

        \item $\phi_i$ consists of all names in $N_P$ and all names introduced in $\act_1\dots\act_{i-1}$.
        
        \item $\xi_i$ contains contexts of the form $(\la x.(!\pqr{j})x)\bullet$ formed by the actions $\act_1\dots\act_{i-1}$ in the natural way, i.e.,
        a (PQ) action pushes a new evaluation context to the top of the stack,
        and an (OA) action pops the evaluation context from the top of the stack.
        To be more precise,
        (a) $(\la x.(!\pqr{0})x)\bullet$ is present in $\xi_i$ for all $i<n$,
        (b) for $1\le j \le n$, $(\la x.(!\pqr{j})x)\bullet$ is present in $\xi_i$
        iff the $j$th (PQ) action appears before $\act_i$ and
        between these two actions the number of (OA) actions is less than or equal to the number of (PQ) actions,
        and (c) the contexts appear in $\xi_i$ according to ascending order of $j$.

        \item $\sigma_i=\dom{\chi_{i-1}}$ for all $i>1$, and $\sigma_1=\emptyset$.
    \end{itemize}

    The remaining task is to define the stored values $h_i(-)$, and $M_i$ where necessary.
    We proceed by backward induction on $i$.
    \begin{itemize}
        \item $i=n+1$. We need $\CC_{n+1}$ to be a passive configuration such that $\Tr{\CC_{n+1}}$ consists only of the empty trace.
        Since $\tr$ is a complete trace, $\xi_{n+1}$ must be empty, so already no OA action is allowed,
        and we can simply set $h_{n+1}(\pqr{j})=\la x.(!\pqr{j})x$ for all $j$.
        We set $h_{n+1}(\pfr{j})=\la x.(!\pfr{j})x$ and $h_{n+1}(\pbr{j})=\bx (\letbox{u}{!\pbr{j}}{u^{id}})$
        for all $j$,
        then all OQ actions will cause divergence.
        Given an abstract value $A$,
        let $V_A=A$ if $A$ is an atomic value,
        and $V_A=\gamma_{n}(A)$ if $A$ is a name.
        Notice that all locations appearing in $\tr$ must appear in $\dom{\chi_n}$,
        so we let $h_{n+1}(\ell)=V_{\chi_{n}(\ell)}$ for all $\ell$.

        \item $i\le n$, and $\act_i$ is a P action.
        We let $h_i(\pfr{j})=h_{i+1}(\pfr{j})$,
        $h_i(\pbr{j})=h_{i+1}(\pbr{j})$,
        and $h_i(\pqr{j})=h_{i+1}(\pqr{j})$ for all $j$.
        We let $h_i(\ell)=\chi_{i-1}(\ell)$ for all $\ell\in\dom{\chi_{i-1}}$,
        and $h_i(\ell)=h_{i+1}(\ell)$ otherwise.

        Let $\mathit{setsharedloc}(i)$ be a sequence of assignments $\ell:=V_A$
        for every $\chi_i(\ell)=A$.
        We define $M_i$ by case analysis on $\act_i$:
        \begin{itemize}
            \item If $\act_i=(\ol{A},\Sigma_i,\chi_i)$,
            then let $M_i$ be the term $\mathit{setsharedloc}(i);V_A$.
            \item If $\act_i=(\ol{f_{O,j}}(A),\Sigma_i,\chi_i)$,
            then let $M_i$ be the term $\mathit{setsharedloc}(i);(\la x.(!\pqr{j'})x)((!\ofr{j}) V_A)$
            if $\act_i$ is the $j'$-th (PQ) action.
            \item If $\act_i=(\ol{\run b_{O,j}^\rho},\Sigma_i,\chi_i)$,
            then let $M_i$ be the term $\mathit{setsharedloc}(i);(\la x.(!\pqr{j'})x)(\letbox{u}{!\obr{j}}{u^{V_\rho}})$
            if $\act_i$ is the $j'$-th (PQ) action,
            where $V_\rho(x)=V_{\rho(x)}$.
        \end{itemize}

        One can verify that indeed $\CC_i\xra{\tau}^*D_i\xra{\act_i}\CC_{i+1}$.

        \item $i\le n$, and $\act_i$ is an O action.
        We let $h_i(\ell)=V_{\chi_{i-1}(\ell)}$ for all $\ell\in\dom{\chi_{i-1}}$,
        and $h_i(\ell)=h_{i+1}(\ell)$ otherwise.
        We define the rest of $h_i(-)$ by case analysis on $\act_i$:
        \begin{itemize}
            \item If $\act_i=(A,\chi_i)$, then we set $h_i(\pfr{j})=\la x.(!\pfr{j})x$ and $h_i(\pbr{j})=\bx (\letbox{u}{!\pbr{j}}{u^{id}})$ for all $j$,
            so that all OQ actions will cause divergence.
            Suppose the context on top of $\xi_i$ is $(\la x.(\pqr{j})x)\bullet$.
            We hope $\pqr{j}$ to perform appropriate operations on receiving $\act_i$,
            so we set $h_i(\pqr{j})=\la x.\ass{x}{A};\mathit{assertshared}(i);\mathit{setauxloc}(i);M_{i+1}$,
            where
            \begin{itemize}
                \item $\ass{x}{A}$ is ``$\iif{(x=A)}{()}{\Omega}$'' if $A$ is a primitive value,
                ``$\ofr{j'}:=x$'' if $A$ is a function name $f_{O,j'}$,
                and ``$\obr{j'}:=x$'' if $A$ is a box name $b_{O,j'}$.

                \item $\mathit{assertshared}(i)$ is the sequence of operations
                ``$\ass{!\ell}{A}$'' for all $\chi_i(\ell)=A$.
                
                \item $\mathit{setauxloc}(i)$ is the sequence of assignments
                ``$\pfr{j'}:=h_{i+1}(\pfr{j'})$'',
                ``$\pbr{j'}:=h_{i+1}(\pbr{j'})$'',
                and ``$\pqr{j'}:=h_{i+1}(\pqr{j'})$'' for all $j'$,
                where $h_{i+1}(-)$ denotes the value that is stored in $h_{i+1}$.
            \end{itemize}
            
            \item If $\act_i=(f_{P,j}(A),\chi_i)$,
            we set $h_i(\pbr{j'})=\bx (\letbox{u}{!\pbr{j'}}{u^{id}})$ and
            $h_i(\pqr{j'})=\la x.(!\pqr{j'})x$ for all $j'$,
            and $h_i(\pfr{j'})=\la x.(!\pfr{j'})x$ for all $j'\neq j$ to disallow all other types of actions.
            We set $h_i(\pfr{j})=\la x.\ass{x}{A};\mathit{assertshared}(i);\mathit{setauxloc}(i);M_{i+1}$ so that it behaves as expected after $\act_i$.
            
            \item If $\act_i=(\run b_{P,j}^\rho,\chi_i)$,
            we set $h_i(\pqr{j'})=\la x.(!\pqr{j'})x$ and
            $h_i(\pfr{j'})=\la x.(!\pfr{j'})x$ for all $j'$.
            We set $h_i(\pbr{j'})=\bx (\letbox{u}{!\pbr{j'}}{u^{id}})$ for all $j'\neq j$,
            and $h_i(\pbr{j})=\bx(\mathit{assertsubst}(i);\mathit{assertshared}(i);\mathit{setauxloc}(i);m_{i+1})$
            where $\mathit{assertsubst}(i)$ is a sequence of operations ``$\ass{x}{A}$'' for all $\rho(x)=A$.
            
        \end{itemize}
        One can verify that indeed $\CC_i\xra{\act_i}E_i\xra{\tau}^*\CC_{i+1}$,
        and there is no other action $\act'$ such that $\CC_i\xra{\act'}E'$.
    \end{itemize}
\end{proof}

Like before, we specialize the above result for some $K,\sigma,\delta,h$.

In a complete $\emptyset,N_P$-trace $\tr$,
there must exist an (OA) action $(\ol{A_1},\Sigma_1,\chi_1)$ that is the first time
the number of (OA) actions equal the number of (PQ) actions plus one.
Suppose $\Sigma_1;\cdot;\cdot\vd_{i_1} A_1:T_1$, then we say $(i_1,T_1)$ is the \textit{O-init type} of this trace.
Similarly, the last action must be a (PA) action $(\ol{A_2},\Sigma_2,\chi_2)$.
Suppose $\Sigma_2;\cdot;\cdot\vd_{i_2} A_2:T_2$, then we say $(i_2,T_2)$ is the \textit{P-init type} of this trace.

\begin{corollary}[Definability]\label{cor:definability}
    Suppose we have a complete $\emptyset,N_P$-trace $\tr$
    with O-init type $(i_1,T_1)$ and P-init type $(i_2,T_2)$.
    There exists $\CC=(\Sigma,h,\gamma,N_P,\cdot : (\Sigma\vd K:(\cdot;\cdot;i_1;T_1)\ra (\cdot;\cdot;i_2;T_2)),\theta)$
    such that $\Tr{\CC}=\{\tr\}$ up to renaming that preserves $N_P$.
\end{corollary}

\begin{proof}
    Follows from \Cref{lem:definability_inductive} by letting $i=1$.
\end{proof}

\begin{theorem}[Completeness, \Cref{thm: completeness}]
    Given two terms $\spgv{i} M_1,M_2:T$,
    if $M_1\lsm M_2$,
    then $\Tr{M_1}\subseteq\Tr{M_2}$.
\end{theorem}

\begin{proof}
    Suppose there exists a global substitution $\Sigma';\cdot\vd \sigma:\Psi$,
    local substitution $\spccv{i} \delta:\Gamma$,
    evaluation context $\spccv{} K: (\cdot;\cdot;i;T)\ra(\cdot;\cdot;i';T')$,
    and heap $\spccv{} h:\heap$
    such that $\Sigma'\supseteq\Sigma$ and $(K[M_1[\sigma][\delta]],h)\Da^0$.
    Also suppose $(\eta,\gamma)\in\AVal{\Sigma';\cdot\vd\sigma:\Psi}$,
    $(\rho,\gamma')\in\AVal{\spccv{i}\delta:\Gamma}$,
    $(\Sigma'',h')=\rst{(\Sigma',h)}{\mu(\eta)\cup\mu(\rho)}$,
    and $(\chi,\gamma'')\in\AVal{\Sigma'';\cdot;\cdot\vd h':\heap}$.
    By \Cref{thm:correctness}, there exists $\tr$ and $\act$
    such that $\tr\in\Tr{\cconf{M_1}{\eta,\rho,\chi}}$ and $\ol{\tr}\; \act\in\Tr{\cconf{K,\sigma,\delta,h}{\eta,\gamma,\rho,\gamma',\chi,\gamma''}}$.
    Since $\Tr{M_1}\subseteq\Tr{M_2}$, we have $\tr\in\Tr{\cconf{M_2}{\eta,\rho,\chi}}$.
    By \Cref{thm:correctness} again, $(K[M_2[\sigma][\delta]],h)\Da^0$.
\end{proof}

\begin{proof}
    Suppose $(\eta,\rho,\chi,\tr)\in\Tr{M_1}$,
    i.e., $\tr\in\Tr{\cconf{M_1}{\eta,\rho,\chi}}$
    for some abstract global substitution $\Sigma';\cdot\vd \eta:\Psi$,
    abstract local substitution $\spccv{i} \rho:\Gamma$,
    and heap $\spccv{} h:\heap$.
    Then $\ol{\tr}\; \act$ is a complete $\emptyset,\nu(\eta)\uplus\nu(\rho)\uplus\nu(\chi)$-trace
    for some arbitrary (PA) action $\act$.
    By \Cref{cor:definability},
    there exists $\CC=(\Sigma',h,\gamma,\nu(\eta)\uplus\nu(\rho)\uplus\nu(\chi),\cdot : (\Sigma'\vd K:(\cdot;\cdot;i;T)\ra (\cdot;\cdot;i';T')),\theta)$
    such that $\Tr{\CC}=\{\tr\}$ up to renamings that preserve $\nu(\rho)\cup\nu(\chi)$.
    Observe that $\CC=\cconf{K,\sigma,\delta,h'}{\eta,\gamma',\rho,\gamma'',\chi,\gamma'''}$
    where $\gamma'=\rst{\gamma}{\nu(\eta)}$,
    $\gamma''=\rst{\gamma}{\nu(\rho)}$,
    $\gamma'''=\rst{\gamma}{\nu(\chi)}$,
    $\sigma(u)=\eta(u)[\gamma]$,
    $\delta(x)=\rho(x)[\gamma]$,
    and $h'(\ell)=h(\ell)[\gamma]$.
    By \Cref{thm:correctness}, $(K[M_1[\sigma][\delta]],h')\Da$.
    Since $M_1\lsm M_2$, $(K[M_2[\sigma][\delta]],h')\Da$.
    By \Cref{thm:correctness} again,
    since $\ol{\tr}\; \ol{A}$ is the only trace (up to renaming) in $\Tr{\cconf{K,\sigma,\delta,h'}{\eta,\gamma',\rho,\gamma'',\chi,\gamma'''}}$,
    we have $\tr\in\Tr{\cconf{M_2}{\eta,\rho,\chi}}$ and therefore $(\eta,\rho,\chi,\tr)\in\Tr{M_2}$.
\end{proof}

\begin{theorem}[Full Abstraction, \Cref{thm:fullAbstraction}]
    Given two terms $\spgv{i} M_1,M_2:T$,
    $M_1\lsm M_2$
    iff $\Tr{M_1}\subseteq\Tr{M_2}$.
\end{theorem}

\begin{proof}
    Follows from \Cref{thm: CIU,thm: soundness,thm: completeness}.
\end{proof}

\subsection{Proofs for \Cref{sec: examples}}

\begin{example}[Correctness of staging the imperative power function, \Cref{exm: correctness staging}]
    Let $\power$ be the following term:
    $$
    \begin{aligned}
        &
        \begin{aligned}
            \vd_1\rec{f}{n}{\la xy.&\mathbf{if}\ n=0\ \mathbf{then}\ {y:=1}\\
            &\mathbf{else}\ (f\ (n-1)\ x\ y;y:=!y*x)}
        \end{aligned}\\
        &\colon\Int\ra \Int\ra\rf \Int\ra \Unit
    \end{aligned}
    $$

    Let $\powerstagedgen$ be the following term:
    $$
    \begin{aligned}
    &
    \begin{aligned}
        \vd_1\mathbf{rec\ }f(n).\la xy.&\mathbf{if}\ n=0\ \mathbf{then}\ \letbox{u}{y}{\bx (u^{y/y}:=1)}\\
        &
        \begin{aligned}
            \mathbf{else}\ &\letbox{u}{f\ (n-1)\ x\ y}\\
            &\letbox{v}{x}\\
            &\letbox{w}{y}\\
            &\bx (u^{x/x,y/y};w^{y/y}:=!w^{y/y}*v^{x/x})
        \end{aligned}
    \end{aligned}\\
    &
    \begin{aligned}
        \colon&\Int\ra \Box(x:\Int\vd\Int)\ra\Box(y:\rf\Int\vd\rf\Int)\\
        &\ra\Box(x:\Int,y:\rf\Int\vd\Unit)
    \end{aligned}
    \end{aligned}
    $$
    Let $\powerstaged$ be the following term:
    $$
    \begin{aligned}
    &\vd_1 \la n.\letbox{u}{\powerstagedgen\ n\ (\bx x)\ (\bx y)}{\la xy.u^{x/x,y/y}} \\
    &\colon\Int\ra\Int\ra\rf\Int\ra\Unit
    \end{aligned}
    $$
    Then $\powerstaged\lnapprox \power\approx \la nxy.\powerstaged\ n\ x\ y$.
\end{example}
\begin{proof}
    It suffices to prove that
    $\Tr{\powerstaged}\subsetneqq\Tr{\power} =$ $\Tr{\la nxy.\powerstaged\ n\ x\ y}$
    by \Cref{thm:fullAbstraction}.
    We first prove $\Tr{\power}=\Tr{\la nxy.\powerstaged\ n\ x\ y}$.
    Consider the representative traces (again $\Sigma$ is omitted for brevity)
    $$
    \begin{aligned}
    \tr_1\ =\ &(\ol{f_1},\cdot)\; (f_1(3),\cdot)\; (\ol{f_2},\cdot)\; (f_2(2),\cdot)\; (\ol{f_3},\cdot)\\
    &(f_3(\ell_1),\bc{\ell_1\mapsto 0})\; (\ol{()},\bc{\ell_1\mapsto 8}) \\
    \end{aligned}
    $$
    and
    $$\tr_2\ =\ (\ol{f_1},\cdot)\; (f_1(-1),\cdot)\; (\ol{f_2},\cdot)$$
    All traces that can be generated from $\power$ and $\la nxy.\powerstaged\ n\ x\ y$
    can be represented as a tree.
    The trace $\tr_1$ is a representative path in the tree,
    and other traces can be obtained by branching at the three (OQ) actions.
    Due to higher-order references, the opponent can also freely return to previous nodes
    and pose questions on the function names with different arguments.
    Any path, as long as it follows the transition rules, is valid up to the action
    of the form $(f_3(\ell_1),\bc{\ell_1\mapsto 0})$
    because no real computation is done until all three arguments are supplied.

    Once $n$, $x$ and $y$ are supplied, both programs will start to do real computation.
    If $n<0$, both programs will diverge.
    If $n\ge 0$, the $\power$ program will correctly compute $x^n$ and store the result in $y$.
    To show that $\la nxy.\powerstaged\ n\ x\ y$ does the same thing,
    it suffices to prove
    that $\powerstagedgen\ n\ (\bx x)\ (\bx y)$
    produces a piece of code that computes $x^n$ and stores the result in $y$.
    We prove this by induction on $n$.
    If $n=0$, then
    $\powerstagedgen\ 0\ (\bx x)\ (\bx y)\ra^*\bx (y:=1)$.
    If $n=k+1$, then
    $\powerstagedgen\ (k+1)\ (\bx x)\ (\bx y)\ra^*\bx(\#(\powerstagedgen\ k\ (\bx x)\ (\bx y)); y := !y * x)$.
    By the induction hypothesis,
    $(\powerstagedgen\ k\ (\bx x)\ (\bx y))$
    computes to a piece of code that stores $x^k$ in $y$.
    Therefore, the whole code stores $x^{k+1}$ in $y$.

    Now we prove $\Tr{\powerstaged}\subsetneqq\Tr{\power}$.
    The traces generated by $\powerstaged$ can be understood in a similar way as above,
    except that not all traces can make it till $(f_3(\ell_1),\bc{\ell_1\mapsto 0})$.
    Since computations related to $n$ are done once $n$ is supplied,
    an invalid $n$ such as $-1$ will cause the program to diverge immediately,
    and no trace of the form $\tr_2$ will be generated.
    Therefore, we have $\Tr{\powerstaged}\subsetneqq\Tr{\power}=\Tr{\la nxy.\powerstaged\ n\ x\ y}$.
\end{proof}

\end{document}